\documentclass[11pt]{article}
\usepackage{ifpdf}
\usepackage{geometry} 
\usepackage[parfill]{parskip}    
\ifpdf
\usepackage[pdftex]{graphicx}
\else
\usepackage[dvips]{graphicx}
\fi
\usepackage{amsmath, amsthm}
\usepackage{amsfonts}
\usepackage{amssymb}
\usepackage{epstopdf}
\usepackage{subfigure}
\usepackage{fancyhdr}
\usepackage{algorithms/algorithmic}
\usepackage{algorithms/algorithm}
\usepackage{array}
\usepackage{url}
\usepackage{comment}

\usepackage{natbib}

\newtheorem{thm}{Theorem}[section]
\newtheorem{cor}[thm]{Corollary}

\newtheorem{claim}[thm]{Claim}

\theoremstyle{definition}

\numberwithin{equation}{section}

\usepackage[usenames,dvipsnames]{color}

\usepackage{bbding}

\ifpdf
\usepackage[pdftex]{hyperref}
\else
\usepackage[dvips]{hyperref}
\fi
\hypersetup{
    bookmarks=true,         
    unicode=false,          
    pdftoolbar=true,        
    pdfmenubar=true,        
    pdffitwindow=false,     
    pdfstartview={FitH},    
    pdftitle={RAICAR-N},    
    pdfauthor={Gautam V. Pendse},     
    pdfsubject={RAICAR-N},   
    pdfcreator={Gautam V. Pendse},   
    pdfproducer={Gautam V. Pendse}, 
    pdfkeywords={ICA} {RSN} {reproducibility}, 
    pdfnewwindow=true,      
    colorlinks=true,       
    linkcolor=red,          
    citecolor=blue,        
    filecolor=magenta,      
    urlcolor=cyan           
}

\newcommand{\vect}[1]{\boldsymbol{#1}}
\newcommand{\matr}[1]{\boldsymbol{#1}}

\newcommand{\mvn}[3]{\mathcal{N}\left( \vect{#1} \mid \vect{#2}, \matr{#3} \right)} 

\title{A simple and objective method for reproducible resting state network (RSN) detection in fMRI}

\author{Gautam V. Pendse\thanks{To whom correspondence should be addressed. e-mail: gpendse@mclean.harvard.edu}$^{\,\,\,1}$, David Borsook$^{\,1,2}$, and Lino Becerra$^{\,1,2}$
\mbox{}\\ \\ \\ \\ 
$^1$ P.A.I.N Group, Imaging and Analysis Group (IMAG), McLean Hospital, Harvard Medical School\\\\
$^2$ A. A. Martinos Center for Biomedical Imaging, Massachusetts General Hospital}
\date{June 15, 2011}

\begin{document}

\maketitle

\newpage 

\begingroup
\hypersetup{linkcolor=red}
\tableofcontents
\endgroup

\section*{Abstract}
Spatial Independent Component Analysis (ICA) decomposes the time by space functional MRI (fMRI) matrix into a set of 1-D basis time courses and their associated 3-D spatial maps that are optimized for mutual independence. When applied to resting state fMRI (rsfMRI), ICA produces several spatial independent components (ICs) that seem to have biological relevance - the so-called resting state networks (RSNs). The ICA problem is well posed when the true data generating process follows a linear mixture of ICs model in terms of the identifiability of the mixing matrix. However, the contrast function used for promoting mutual independence in ICA is dependent on the finite amount of observed data and is potentially non-convex with multiple local minima. Hence, each run of ICA could produce potentially different IC estimates even for the same data. One technique to deal with this run-to-run variability of ICA was proposed by \citet{Yang:2008} in their algorithm RAICAR which allows for the selection of only those ICs that have a high run-to-run reproducibility. We propose an enhancement to the original RAICAR algorithm that enables us to assign reproducibility $p$-values to each IC and allows for an objective assessment of both within subject and across subjects reproducibility. We call the resulting algorithm RAICAR-N (N stands for null hypothesis test), and we have applied it to publicly available human rsfMRI data (\url{http://www.nitrc.org}). Our reproducibility analyses indicated that many of the published RSNs in rsfMRI literature are highly reproducible. However, we found several other RSNs that are highly reproducible but not frequently listed in the literature.

\section*{Notation}
\begin{itemize}
\item Scalars variables and functions will be denoted in a non-bold font (e.g., $\sigma^2, L, p$ or $\Psi, f$). Vectors will be denoted in a bold font using lower case letters (e.g., $\vect{y}, \vect{\mu}, \vect{\eta}$). Matrices will be denoted in bold font using upper case letters (e.g., $\matr{A}, \matr{\Sigma}, \matr{W}$). The transpose of a matrix $\matr{A}$ will be denoted by $\matr{A^T}$ and its inverse will be denoted by $\matr{A^{-1}}$. $\matr{I_p}$ will denote the $p \times p$ identity matrix and $\mathbf{0}$ will denote a vector or matrix of all zeros whose size should be clear from context. ${N \choose L}$ is the number of ways of choosing $L$ objects from $N$ objects when order does not matter.

\item The $j$th component of vector $\vect{t_i}$ will be denoted by $t_{ij}$ whereas the $j$th component of vector $\vect{t}$ will be denoted by $t_{j}$. The element $(i,j)$ of matrix $\matr{G}$ will be denoted by $G(i,j)$ or $G_{ij}$. Estimates of variables will be denoted by putting a hat on top of the variable symbol. For example, an estimate of $\vect{s}$ will be denoted by $\vect{\hat{s}}$.

\item If $\vect{x}$ is a random vector with a multivariate Normal distribution with mean $\vect{\mu}$ and covariance $\matr{\Sigma}$ then we will denote this distribution by $\mvn{x}{\mu}{\Sigma}$. The joint density of vector $\vect{s}$ will be denoted by $\vect{p}_{\vect{s}}( \vect{s} )$ whereas the marginal density of $s_i$ will be denoted as $p_{s_i}( s_i )$. $\mathbf{E} \left[ f(\vect{s}, \vect{\eta} ) \right]$ denotes the expectation of $f(\vect{s}, \vect{\eta})$ with respect to both random variables $\vect{s}$ and $\vect{\eta}$.
\end{itemize}

\section{Introduction}

Independent component analysis (ICA) \citep{Jutten:1991, Comon:1994, Bell:1995, Attias:1999} models the observed data as a linear combination of a set of statistically independent and unobservable sources. \citep{McKeown:1998} first proposed the application of ICA to the analysis of functional magnetic resonance imaging (fMRI) data. Subsequently, ICA has been applied to fMRI both as an exploratory tool for the purpose of identifying task related components \citep{McKeown:1998} as well as a signal clean up tool for the purpose of removing artifacts from the fMRI data \citep{Tohka:2008}. Recently, it has been shown that ICA applied to resting state fMRI (rsfMRI) in healthy subjects reveals a set of biologically meaningful spatial maps of independent components (ICs) that are consistent across subjects - the so called resting state networks (RSNs) \citep{Beckmann:2005}. Hence, there is a considerable interest in applying ICA to rsfMRI data in order to define the set of RSNs that characterize a particular group of human subjects, a disease, or a pharmacological effect.

Several variants of the linear ICA model have been applied to fMRI data including square ICA (with equal number of sources and sensors) \citep{McKeown:1998}, non-square ICA (with more sensors than sources) \citep{Calhoun:2001}, and non-square ICA with additive Gaussian noise (noisy ICA) \citep{Beckmann:2004}. All of these models are well known in the ICA literature \citep{Jutten:1991, Cardoso:1998, Comon:1994, Attias:1999}. Since the other ICA models are specializations of the noisy ICA model, we will assume a noisy ICA model henceforth. 

Remarkably, the ICA estimation problem is well posed in terms of the identifiability of the mixing matrix given several non-Gaussian and at most 1 Gaussian source in the overall linear mixture \citep{Rao:1969, Comon:1994, Theis:2004, Davies:2004}. In the presence of more than 1 Gaussian source, such as in noisy ICA, the mixing matrix corresponding to the non-Gaussian part of the linear mixture is identifiable (upto permutation and scaling). In addition, the source distributions are uniquely identifiable (upto permutation and scaling) given a noisy ICA model with a particular Gaussian co-variance structure, for example, the isotropic diagonal co-variance. 
For details, see section \ref{identifiability}.

While these uniqueness results are reassuring, a number of practical difficulties prevent the reliable estimation of ICs on real data. These difficulties include (1) true data not describable by an ICA model, (2) ICA contrast function approximations, (3) multiple local minima in the ICA contrast function, (4) confounding Gaussian noise and (5) model order overestimation. See section \ref{run-to-run} for more details. A consequence of these difficulties is that multiple ICA runs on the same data or different subsets of the data produce different estimates of the IC realizations.

One technique to account for this run-to-run variability in ICA was proposed by \citep{Himberg:2004} in their algorithm ICASSO. Using repeated runs of ICA with bootstrapped data using various initial conditions, ICASSO clusters ICs across ICA runs using agglomerative hierarchical clustering and also helps in visualizing the estimated ICs. The logic is that reliable ICs will show up in almost all ICA runs and thus will form a tight cluster well separated from the rest. \citep{Esposito:2005} proposed a technique similar to ICASSO called self-organizing group ICA (sogICA) which allows for clustering of ICs via hierarchical clustering in across subject ICA runs. When applied to multiple ICA runs across subjects, ICASSO does not restrict the IC clusters to contain only 1 IC from each subject per ICA run. In contrast, sogICA allows the user to select the minimum number of subjects for a "group representative" IC cluster containing distinct subjects. By labelling each ICA run as a different "subject" sogICA can also be applied to analyze multiple ICA runs across subjects.

Similar in spirit to ICASSO and sogICA, \citet{Yang:2008} proposed an intuitive approach called RAICAR (Ranking and Averaging Independent Component Analysis by Reproducibility) for reproducibility analysis of estimated ICs. The basic idea in RAICAR is to select only those ICs as "interesting" or "stable" which show a high run-to-run "reproducibility". RAICAR uses simple and automated spatial cross-correlation matrix based IC alignment which has been shown to be more accurate compared to ICASSO \citep{Yang:2008}. RAICAR is applicable to both within subject as well as across subjects reproducibility analysis.

A few limitations of ICASSO, sogICA and RAICAR are worth noting:
\begin{itemize}
\item ICASSO requires the user to select the number of IC clusters and is inapplicable without modification for across subjects analysis of ICA runs since the IC clusters are not restricted to contain only 1 IC per ICA run. 
\item sogICA requires the user to select the minimum number of subjects for a "group representative" cluster and also a cutoff on within cluster distances. 
\item RAICAR uses an arbitrary threshold on the reproducibility indices selected "by eye" or set at an arbitrary value, such as $50\%$ of the maximum reproducibility value. 
\end{itemize}
We propose a simple extension to RAICAR that avoids subjective user decisions and allows for an automatic reproducibility cutoff. The reproducibility indices calculated in RAICAR differ in magnitude significantly depending on whether the input to RAICAR: 
\begin{itemize}
\item (a) is generated using multiple ICA runs on the same data
\item (b) comes from multiple ICA runs on varying data sets (e.g. between and across subject runs)
\end{itemize}

\begin{figure}[htbp]
\begin{center}
\includegraphics[width=4.5in]{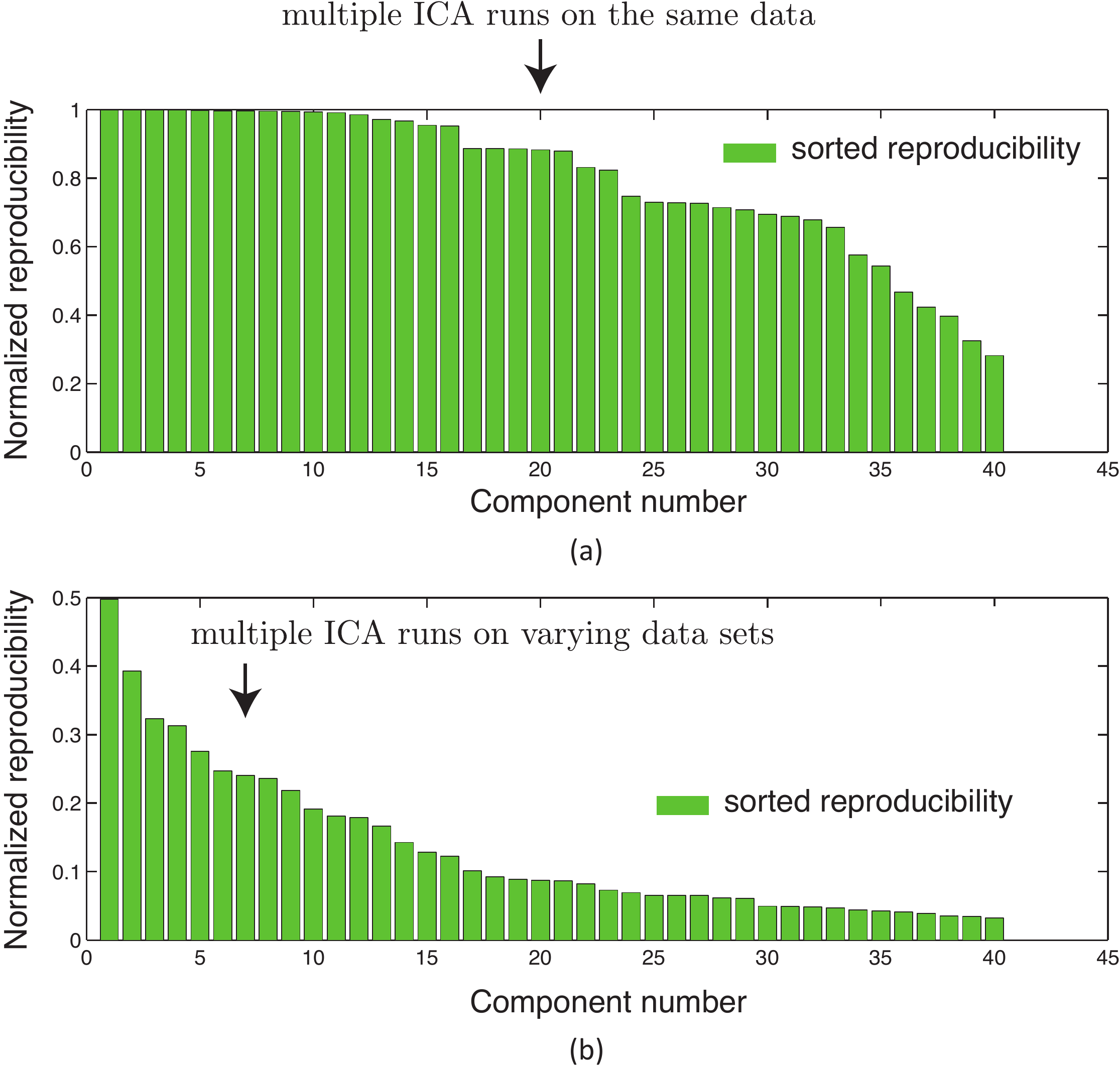}
\caption{Figure illustrates the variation in normalized reproducibility from RAICAR depending on whether the input to RAICAR is (a) Multiple ICA runs on single subject data or (b) Multiple ICA runs across subjects. Notice that the normalized reproducibility is much lower for across subjects analysis compared to within subject analysis.}
\label{figure1}
\end{center}
\end{figure}

See Figure \ref{figure1} for an illustration of this effect. Obviously, the reproducibility indices are much lower in case (b) since we account for both within subject and between subjects variability in estimating ICs. Case (b) is also of great interest from a practical point of view since we are often interested in making statements about a group of subjects. Hence, it is clear that a cutoff on RAICAR reproducibility values for the purposes of selecting the "highly reproducible" components should be data dependent. In this work, 

\begin{enumerate}
\item We propose a modification of the original RAICAR algorithm by introducing an explicit "null" model of no reproducibility. 
\item We use this "null" model to automatically generate $p$-values for each IC via simulation. This allows for an objective cutoff specification for extracting reproducible ICs (e.g. reproducible at $p < 0.05$) within and across subjects. We call the resulting algorithm RAICAR-N (N stands for "null" hypothesis test). 
\item We validate RAICAR-N by applying it to publicly available human rsfMRI data. 
\end{enumerate}

\section{Methods}
The organization of this article revolves around the following sequence of questions which ultimately lead to the development of RAICAR-N:
\begin{enumerate}
\item Why is a reproducibility assessment necessary in ICA analysis? In order to answer this question, we cover the fundamentals of ICA including identifiability issues in sections \ref{background} and \ref{ica_single_subject_and_group}.

\item How does the original RAICAR algorithm assess reproducibility? The answer to this question in section \ref{orig_raicar} will set up the stage for RAICAR-N.

\item How does RAICAR-N permit calculation of reproducibility $p$-values? In section \ref{raicar_n}, we describe the RAICAR-N "null" model and a simulation based approach for assigning $p$-values to ICs.

\item How to promote diversity in group ICA runs given a limited number of subjects when using RAICAR-N and how to display the non-Gaussian spatial structure in estimated ICs? These issues are covered in section \ref{choosingL} and \ref{non_gaussian_structure}.

\item How can RAICAR-N be extended for between group comparison of ICs and how does it compare to other approaches in the literature? This question is addressed in section \ref{raicar-n-across-groups}.
\end{enumerate}

\subsection{ICA background}\label{background}

In this section, we provide a brief introduction to ICA along with a discussion of associated issues related to model order selection, identifiability and run-to-run variability. The noisy ICA model assumes that observed data $\vect{y}$ is generated as a linear combination of unobservable independent sources confounded with Gaussian noise:
\begin{equation}\label{eq1d}
\vect{y} = \vect{\mu} + \matr{A} \, \vect{s} + \vect{\eta}
\end{equation}

In this model,
\begin{align}\label{eq2d}
\vect{y} &= p \times 1 \mbox{ observed signal vector} \\\nonumber
\vect{\mu} &= p \times 1 \mbox{ mean vector} \\\nonumber
\matr{A} &= p \times q \mbox{ mixing matrix  with $p > q$ (more sensors than sources) and rank $q$ } \\\nonumber
\vect{\eta} &= p \times 1 \mbox{ Gaussian noise vector with density } \mvn{\eta}{\vect{0}}{\matr{\Sigma}} \\\nonumber
\vect{s} &= q \times 1 \mbox{ vector of independent random variables (the ICs) } \\ \nonumber
&\mbox{ $\,\,\,\,$ with }\mathbf{E}( \vect{s}\vect{s}^T ) = \matr{D} \mbox{ (diagonal)} \mbox{ and } \mathbf{E}(\vect{s}) = \matr{0} \\ \nonumber
&\mbox{ $\,\,\,\,$ and with } \vect{s} \mbox{ and } \vect{\eta} \mbox{ independent }
\end{align}
If the marginal density of the $i$th source $s_i$ is $p_{s_i} (s_i)$ then the joint source density $\vect{p}_{\vect{s}}(\vect{s})$ factorizes as $\prod_{i=1}^q p_{s_i} (s_i)$ because of the independence assumption but is otherwise assumed to be unknown. Also, since the elements of $\vect{s}$ are independent their co-variance matrix $\matr{D}$ is diagonal. The set of variables $\mathcal{F} = \left\{ \vect{\mu}, \matr{A}, \matr{D}, \matr{\Sigma} \right\}$ represents the unknown parameters in the noisy ICA model. Before discussing the identifiability of model \ref{eq1d}, we briefly discuss the choice of model order or the assumed number of ICs $q$.
 
\subsubsection{Estimating the model order $q$}

Rigorous estimation of the model order $q$ in noisy ICA is difficult as the IC densities $p_{s_i} (s_i)$ are unknown. This means that $\vect{p} \left(\vect{y} \mid q, \mathcal{F} \right)$, the marginal density of the observed data given the model order and the ICA parameters cannot be derived in closed form (by integrating out the ICs) without making additional assumptions on the form of IC densities. Consequently, standard model selection criteria such as Bayes information criterion (BIC) \citep{Kass:1993_techrep} cannot be easily applied to the noisy ICA model to estimate $q$. One solution is to use a factorial mixture of Gaussians (MOG) joint source density model as in \citep{Attias:1999}, and use the analytical expression for $\vect{p} \left(\vect{y} \mid q, \mathcal{F} \right)$ in conjunction with BIC. This solution is quite general in terms of allowing for an arbitrary Gaussian noise co-variance $\matr{\Sigma}$, but maximizing $\vect{p} \left(\vect{y} \mid q, \mathcal{F} \right)$ with respect to $\mathcal{F}$ becomes computationally intractable using an expectation maximization (EM) algorithm for $q > 13$ ICs \citep{Attias:1999}. Another rigorous non-parametric approach for estimating $q$ that is applicable to the noisy ICA model with isotropic diagonal Gaussian noise co-variance i.e., with $\matr{\Sigma} = \sigma^2 \matr{I_p}$ is the random matrix theory based sequential hypothesis testing approach of \citet{Kritchman:2009}. To the best of our knowledge, these are the only 2 rigorous approaches for estimating $q$ in the noisy ICA model.

Approximate approaches for estimating $q$ commonly used in fMRI literature (e.g., \citep{Beckmann:2004}) consist of first relaxing the isotropic diagonal noisy ICA model (with $\matr{\Sigma} = \sigma^2 \matr{I_p}$) into a probabilistic PCA (PPCA) model of \citep{Tipping:1999} where the source densities are assumed to be Gaussian i.e., where $\vect{p}_{\vect{s}}(\vect{s}) = \mvn{\vect{s}}{\vect{0}}{\matr{I}_q}$. When using the PPCA model, it becomes possible to integrate out the Gaussian sources to get an expression for $\vect{p} \left(\vect{y} \mid q, \mathcal{F} \right)$ that can be analytically maximized \citep{Tipping:1999}. Subsequently, methods such as BIC can be applied to estimate $q$. Alternative approaches for estimating $q$ in the PPCA model consist of the Bayesian model selection of \citet{Minka:2000}, or in data-rich situations such as fMRI, even the standard technique of cross-validation \citep{Hastie:book}.

From a biological point of view, it has been argued \citep{Cole:2010} that the number of extracted ICs simply reflect the various \textit{equally valid} views of the human functional neurobiology - smaller number of ICs represent a coarse view while a larger number of ICs represent a more fine grained view. However, it is worth noting that from a statistical point of view, over-specification of $q$ will lead to over-fitting of the ICA model which might render the estimated ICs less generalizable across subjects. On the other hand, under-specification of $q$ will result in incomplete IC separation. Both of these scenarios are undesirable.

\subsubsection{Identifiability of the noisy ICA model}\label{identifiability}

To what extent is the noisy linear ICA model identifiable?
Consider a potentially different decomposition of the noisy ICA model \ref{eq1d}:
\begin{equation}\label{eq3d}
\vect{y} = \vect{\mu_1} + \matr{A_1} \, \vect{s_1} + \vect{\eta_1}
\end{equation}
where
\begin{align}\label{eq4d}
\vect{y} &= p \times 1 \mbox{ observed signal vector} \\\nonumber
\vect{\mu_1} &= p \times 1 \mbox{ mean vector} \\\nonumber
\matr{A_1} &= p \times q \mbox{ mixing matrix  with $p > q$ (more sensors than sources) and rank $q$ } \\\nonumber
\vect{\eta_1} &= p \times 1 \mbox{ Gaussian noise vector with density } \mvn{\eta_1}{\vect{0}}{\matr{\Sigma_1}} \\\nonumber
\vect{s_1} &= q \times 1 \mbox{ vector of independent random variables (the ICs) } \\ \nonumber
&\mbox{ $\,\,\,\,$ with }\mathbf{E}( \vect{s}\vect{s}^T ) = \matr{D_1} \mbox{ (diagonal)} \mbox{ and } \mathbf{E}(\vect{s}) = \matr{0} \\\nonumber
&\mbox{ $\,\,\,\,$ and with } \vect{s_1} \mbox{ and } \vect{\eta_1} \mbox{ independent }
\end{align}
What can be said about the equivalence between the parameterizations in \ref{eq1d} and \ref{eq3d}? 
\paragraph{Identifiability of $\vect{\mu}$:}
Equating the expectations of the right hand size of \ref{eq3d} and \ref{eq1d} and noting that $\vect{s},\vect{\eta}, \vect{s_1}, \vect{\eta_1}$ have mean $\vect{0}$ we get:
\begin{equation}\label{eq5d}
\vect{\mu_1} = \vect{\mu}
\end{equation}
Thus the mean vector $\vect{\mu}$ is exactly identifiable. 
\paragraph{Identifiability of $\matr{A}$:}
A fundamental decomposition result states that the noisy ICA problem is well-posed in terms of the identifiability of the mixing matrix $\matr{A}$ upto permutation and scaling provided that the components of $\vect{s}$ are independent and \textbf{non-Gaussian} \citep{Rao:1969, Comon:1994, Theis:2004, Davies:2004}. If $\matr{\Lambda}$ is a diagonal scaling matrix and $\matr{P}$ is a permutation matrix then the identifiability result can be stated as:
\begin{equation}\label{eq6d}
\matr{A_1} = \matr{A} \, \matr{\Lambda} \, \matr{P}
\end{equation}
where \ref{eq3d} is another decomposition of $\vect{y}$ with $\vect{s_1}$ containing independent and non-Gaussian components. In other words, the mixing matrix $\matr{A}$ is identifiable upto permutation and scaling.

\paragraph{Identifiability of $\matr{D}$ and $\matr{\Sigma}$:}
Equating the second moments of the right hand side of \ref{eq3d} and \ref{eq1d} and noting the equality of means \ref{eq5d} and the independence of $\vect{s}, \vect{\eta}$ and $\vect{s_1}, \vect{\eta_1}$ we get:
\begin{equation}\label{eq7d}
\matr{E}\left[ (\vect{y} - \vect{\mu}) (\vect{y} - \vect{\mu})^T \right] = \matr{A} \matr{D} \matr{A^T} + \matr{\Sigma} = \matr{A_1} \matr{D_1} \matr{A_1^T} + \matr{\Sigma_1}
\end{equation}

Let $\matr{W}$ be a $q \times p$ matrix and $\matr{\tilde{Q}}$ be a $p \times (p-q)$ orthogonal matrix such that:
\begin{align}\label{eq8d}
\matr{W} &= ( \matr{A^T} \matr{A} )^{-1} \matr{A^T} \\ \nonumber
\matr{\tilde{Q}^T} \matr{A} &= \matr{0} \\ \nonumber
\matr{\tilde{Q}^T} \matr{\tilde{Q}} &= \matr{I_{p-q}} 
\end{align}
From \ref{eq8d} and \ref{eq7d} we get:
\begin{align}\label{eq9d}
\matr{D} + \matr{W} \matr{\Sigma} \matr{W^T} &= \matr{\Lambda} \, \matr{P} \matr{D_1} \matr{P^T} \matr{\Lambda^T} + \matr{W} \matr{\Sigma_1} \matr{W^T} \\ \nonumber
\matr{\tilde{Q}^T} \matr{\Sigma} \matr{\tilde{Q}} &= \matr{\tilde{Q}^T} \matr{\Sigma_1} \matr{\tilde{Q}} 
\end{align}

\textbf{Case 1: $\matr{\Sigma} = \sigma^2 \matr{I_p}$ and $\matr{\Sigma_1} = \sigma_1^2 \matr{I_p}$} \label{isotropic}\\
The second equation in \ref{eq9d} along with the orthogonality of $\matr{\tilde{Q}}$ gives $\sigma^2 = \sigma_1^2$ and thus $\matr{\Sigma} = \matr{\Sigma_1}$. If we fix the scaling of $\matr{A_1}$ by selecting $\matr{\Lambda^2} = \matr{I_q}$ then from the first equation in \ref{eq9d} we get:
\begin{align}\label{eq9d_1} 
\matr{D} &= \matr{\Lambda} \, \matr{P} \matr{D_1} \matr{P^T} \matr{\Lambda^T} \\ \nonumber
                &=  \matr{P} \matr{D_1} \matr{P^T} \matr{\Lambda^2}  && \text{($\matr{P} \matr{D_1} \matr{P^T}$ is diagonal)} \\ \nonumber 
                &= \matr{P} \matr{D_1} \matr{P^T} 
\end{align}
 In other words, the noise co-variance $\matr{\Sigma} = \sigma^2 \matr{I_p}$ is uniquely determined and for a fixed scaling $\matr{\Lambda^2} = \matr{I_q}$, the source variances $\matr{D}$ are also uniquely determined upto permutation.

\textbf{Case 2: $\matr{\Sigma}$ and $\matr{\Sigma_1}$ arbitrary positive definite matrices}\\
Suppose $\matr{X}$ is a square matrix and let $\mbox{diag}(\matr{X})$ be the diagonal matrix obtained by setting the non-diagonal elements of $\matr{X}$ to $0$ and similarly let $\mbox{offdiag}(\matr{X})$ be the matrix obtained by setting the diagonal elements of $\matr{X}$ to 0. The noise-covariance is partially identifiable by the following conditions:
\begin{align}\label{eq10d}
\matr{\tilde{Q}^T} \matr{\Sigma} \matr{\tilde{Q}} &= \matr{\tilde{Q}^T} \matr{\Sigma_1} \matr{\tilde{Q}}  \\ \nonumber
\mbox{offdiag} \left( \matr{W} \matr{\Sigma} \matr{W^T} \right) &= \mbox{offdiag} \left( \matr{W} \matr{\Sigma_1} \matr{W^T} \right) 
\end{align}
For a fixed scaling $\matr{\Lambda^2} = \matr{I_q}$, the sources variances $\matr{D}, \matr{D_1}$ are constrained by:
\begin{equation}\label{eq11d}
\matr{D} + \mbox{diag}( \matr{W} \matr{\Sigma} \matr{W^T} ) = \matr{P} \matr{D_1} \matr{P^T}  + \mbox{diag}( \matr{W} \matr{\Sigma_1} \matr{W^T}  ) \\
\end{equation}

In general, the source variances $\matr{D}$ cannot be uniquely determined as noted in \citep{Davies:2004}.

\paragraph{Identifiability of the distribution of $\vect{s}$:}
Is the distribution of the non-Gaussian components of $\vect{s}$ identifiable? From \ref{eq1d} and \ref{eq3d}:
\begin{equation}\label{eq12d}
\vect{\mu} + \matr{A} \, \vect{s} + \vect{\eta} = \vect{\mu_1} + \matr{A_1} \, \vect{s_1} + \vect{\eta_1}
\end{equation}
Substituting \ref{eq5d} and \ref{eq6d} in \ref{eq12d} we get:
\begin{equation}\label{eq13d}
\matr{A} \, \vect{s} + \vect{\eta} = \matr{A} \, \matr{\Lambda} \, \matr{P} \, \vect{s_1} + \vect{\eta_1}
\end{equation} 
Premultiplying both sides by $\matr{W}$ from \ref{eq8d} we get:
\begin{equation}\label{eq14d}
\vect{s} + \matr{W} \vect{\eta} = \matr{\Lambda} \, \matr{P} \, \vect{s_1} + \matr{W} \vect{\eta_1}
\end{equation}
Let $\Psi_{\vect{s}}, \Psi_{\vect{\eta}}, \Psi_{\vect{s_1}}, \Psi_{\vect{\eta_1}}$ be the characteristic functions of $\vect{s}, \vect{\eta}, \vect{s_1}$ and $\vect{\eta_1}$ respectively. Then
\begin{align}\label{eq15d}
\Psi_{\vect{s}}( \vect{t} ) &= \mathbf{E}\left[ \mbox{exp} \left(i \vect{t^T} \vect{s}\right) \right]  & \Psi_{\vect{\eta}}( \vect{t} ) &= \mathbf{E}\left[ \mbox{exp} \left(i \vect{t^T} \vect{\eta}\right) \right] \\ \nonumber
\Psi_{\vect{s_1}}( \vect{t} ) &= \mathbf{E}\left[ \mbox{exp} \left(i \vect{t^T} \vect{s_1}\right) \right] & \Psi_{\vect{\eta_1}}( \vect{t} ) &= \mathbf{E}\left[ \mbox{exp} \left(i \vect{t^T} \vect{\eta_1}\right) \right] \\ \nonumber
\end{align}
where $i = \sqrt{-1}$ and $\vect{t}$ is a vector of real numbers of length equal to that of the corresponding random vectors in \ref{eq15d}. Using \ref{eq14d}, we can write:
\begin{align}\label{eq16d}
\mathbf{E} \left[  \mbox{exp} \left( i \vect{t^T} \left\{ \vect{s} + \matr{W} \vect{\eta} \right\} \right) \right] = \mathbf{E} \left[ \mbox{exp} \left( i \vect{t^T} \left\{\matr{\Lambda} \, \matr{P} \, \vect{s_1} + \matr{W} \vect{\eta_1}\right\} \right) \right] \mbox{ for all } \vect{t} \in \mathbf{R^q}
\end{align}
Noting the independence of $\vect{s}, \vect{\eta}$ and $\vect{s_1}, \vect{\eta_1}$:
\begin{align}\label{eq17d}
\mathbf{E} \left[  \mbox{exp} \left( i \vect{t^T} \left\{ \vect{s} \right\} \right) \right] \, \mathbf{E} \left[  \mbox{exp} \left( i \vect{t^T} \left\{ \matr{W} \vect{\eta} \right\} \right) \right] &= \mathbf{E} \left[ \mbox{exp} \left( i \vect{t^T} \left\{\matr{\Lambda} \, \matr{P} \, \vect{s_1} \right\} \right) \right] \, \mathbf{E} \left[ \mbox{exp} \left( i \vect{t^T} \left\{ \matr{W} \vect{\eta_1}\right\} \right) \right] \\ \nonumber
\Rightarrow \Psi_{\vect{s}} \left( \vect{t} \right) \,\, \Psi_{\vect{\eta}} \left( \matr{W^T} \vect{t} \right) &= \Psi_{\vect{s_1}} \left( \matr{P^T} \matr{\Lambda^T} \vect{t} \right) \,\, \Psi_{\vect{\eta_1}} \left( \matr{W^T} \vect{t} \right) \mbox{ for all } \vect{t} \in \mathbf{R^q}
\end{align}
Now $\vect{\eta}$ and $\vect{\eta_1}$ are multivariate Gaussian random vectors both with mean $\vect{0}$ and co-variance matrix $\matr{\Sigma}$ and $\matr{\Sigma_1}$ respectively. Hence, their characteristic functions are given by \citep{Feller:book, Bryc:1995}:
\begin{align}\label{eq18d}
\Psi_{\vect{\eta}} \left( \matr{W^T} \vect{t} \right) = \mbox{ exp } \left( -\frac{1}{2} \vect{t^T} \matr{W} \matr{\Sigma} \matr{W^T} \vect{t} \right) \mbox{ for all } \vect{t} \in \mathbf{R^q} \\\nonumber
\Psi_{\vect{\eta_1}} \left( \matr{W^T} \vect{t} \right) = \mbox{ exp } \left( -\frac{1}{2} \vect{t^T} \matr{W} \matr{\Sigma_1} \matr{W^T} \vect{t} \right) \mbox{ for all } \vect{t} \in \mathbf{R^q}
\end{align}

\begin{claim}\label{source_identifiability}
A sufficient condition for identifiability upto permutation and scaling of the non-Gaussian distributions in $\vect{s}$ given two different parameterizations in \ref{eq1d} and \ref{eq3d} is:
\begin{align}\label{eq19d}
\mbox{diag}( \matr{W} \matr{\Sigma} \matr{W^T} ) = \mbox{diag} ( \matr{W} \matr{\Sigma_1} \matr{W^T} )
\end{align}
\end{claim}
\begin{proof}
From \ref{eq19d} and \ref{eq10d}, we get:
\begin{align}\label{eq20d}
\matr{W} \matr{\Sigma} \matr{W^T} =  \matr{W} \matr{\Sigma_1} \matr{W^T}
\end{align}
Thus from \ref{eq18d},
\begin{align}\label{eq21d}
\Psi_{\vect{\eta}} \left( \matr{W^T} \vect{t} \right) = \Psi_{\vect{\eta_1}} \left( \matr{W^T} \vect{t} \right) \mbox{ for all } \vect{t} \in \mathbf{R^q}
\end{align}
From \ref{eq18d}, $\Psi_{\vect{\eta}} \left( \matr{W^T} \vect{t} \right)$ and $\Psi_{\vect{\eta_1}} \left( \matr{W^T} \vect{t} \right)$ are not equal to 0 for any finite $\vect{t}$, therefore, from \ref{eq21d} and \ref{eq17d} we get:
\begin{align}\label{eq22d}
\Psi_{\vect{s}} \left( \vect{t} \right) \,\, &= \Psi_{\vect{s_1}} \left( \matr{P^T} \matr{\Lambda^T} \vect{t} \right) \mbox{ for all } \vect{t} \in \mathbf{R^q}
\end{align}
Note that $\matr{\Lambda}$ is a diagonal scaling matrix with entries $\lambda_1, \lambda_2, \ldots, \lambda_q$ on the diagonal and $\matr{P}$ is a permutation matrix. Thus, 
\begin{align}\label{eq23d}
\matr{P}^T \matr{\Lambda^T} \vect{t} = \begin{pmatrix} \lambda_{i_1} t_{i_1} \\ \lambda_{i_2} t_{i_2} \\ \vdots \\ \lambda_{i_q} t_{i_q} \end{pmatrix}
\end{align}
where $i_1, i_2, \ldots, i_q$ is some permutation of integers $1,2,\ldots,q$. Suppose $\Psi_{\vect{s}(j)}$ is the characteristic function of the $j$th component of $\vect{s}$ and $\Psi_{\vect{s_1}(j)}$ is the characteristic function of the $j$th component of $\vect{s_1}$. Since the components of $\vect{s}$ and $\vect{s_1}$ are independent by assumption, the joint characteristic functions $\Psi(\vect{s})$ and $\Psi(\vect{s_1})$ factorize:

\begin{align}\label{eq24d}
\Psi_{\vect{s}} \left( \vect{t} \right) &= \Psi_{\vect{s}(1)}(t_1) \, \Psi_{\vect{s}(2)}(t_2) \ldots \Psi_{\vect{s}(j)}(t_j) \ldots \Psi_{\vect{s}(q)}(t_q) \\ \nonumber
\Psi_{\vect{s_1}} \left( \matr{P^T} \matr{\Lambda^T} \vect{t} \right) &=  \Psi_{\vect{s_1}(1)}( \lambda_{i_1} t_{i_1} ) \Psi_{\vect{s_1}(2)}( \lambda_{i_2} t_{i_2} ) \ldots \Psi_{\vect{s_1}(j)}( \lambda_{i_j} t_{i_j} ) \ldots \Psi_{\vect{s_1}(q)}( \lambda_{i_q} t_{i_q} )  
\end{align}

From \ref{eq24d} and \ref{eq22d}
\begin{align}\label{eq25d}
\Psi_{\vect{s}(1)}(t_1) \ldots \Psi_{\vect{s}(j)}(t_j) \ldots \Psi_{\vect{s}(q)}(t_q) &= \Psi_{\vect{s_1}(1)}( \lambda_{i_1} t_{i_1} ) \ldots \Psi_{\vect{s_1}(j)}( \lambda_{i_j} t_{i_j} ) \ldots \Psi_{\vect{s_1}(q)}( \lambda_{i_q} t_{i_q} )  
\end{align}

All characteristic functions satisfy \citep{Feller:book, Bryc:1995}:
\begin{align}\label{eq26d}
\Psi_{\vect{s}(k)}(0) &= 1\\ \nonumber
\Psi_{\vect{s_1}(k)}(0) &= 1 \mbox{ for all $k$ }
\end{align}

Since $i_1,i_2,\ldots,i_q$ is simply a permutation of integers $1,2,\ldots,q$, there exists a $j$ such that $i_j = 1$. Then set $t_2 = 0, t_3 = 0, \ldots, t_q = 0$ in \ref{eq25d}. Then \ref{eq26d} and \ref{eq25d} imply:
\begin{align}\label{eq27d}
\Psi_{\vect{s}(1)}(t_1) &= \Psi_{\vect{s_1}(j)}( \lambda_{i_j} t_{i_j} ) =  \Psi_{\vect{s_1}(j)}( \lambda_{1} t_{1} ) \mbox{ for all } t_1 \in \mathbf{R}
\end{align}

Select the scaling matrix as $\matr{\Lambda^2} = \matr{I_q}$ and thus $\matr{\Lambda}$ is a diagonal matrix with elements $\pm 1$ on the diagonal. Thus $\lambda_1 = \pm 1$ and \ref{eq27d} can be re-written as:
\begin{align}\label{eq28d}
\Psi_{\vect{s}(1)}(t_1) =  \Psi_{\vect{s_1}(j)}( \pm t_{1} ) \mbox{ for all } t_1 \in \mathbf{R} 
\end{align}
Therefore,
\begin{align}\label{eq28d_1}
\Psi_{\vect{s}(1)}(t_1) &=  \Psi_{\vect{s_1}(j)}( t_{1} ) \mbox{ for all } t_1 \in \mathbf{R} \\ \nonumber
\mbox{or} &\\ \nonumber
\Psi_{\vect{s}(1)}(t_1) &=  \Psi_{\vect{s_1}(j)}( - t_{1} ) = \Psi_{-\vect{s_1}(j)}( t_{1} ) \mbox{ for all } t_1 \in \mathbf{R}
\end{align}

Hence the characteristic function of the $1$st component of $\vect{s}$ is identical to the characteristic function of the (possibly sign-flipped) $j$th component of $\vect{s_1}$. Since characteristic functions uniquely characterize a probability distribution \citep{Feller:book}, the distribution of $\vect{s}(1)$ and $\pm \vect{s_1}(j)$ is identical. Next, by setting $t_1 = 0, t_3 = 0, \ldots, t_q = 0$, we can find a distribution from $\vect{s_1}$ that matches the $2$nd component $\vect{s}(2)$ of $\vect{s}$. Proceeding in a similar fashion, it is clear that the distribution of each component of $\vect{s}$ is uniquely identifiable upto sign flips for the choice $\matr{\Lambda^2} = \matr{I_q}$. For a general $\matr{\Lambda}$, the source distributions are uniquely identifiable upto permutation and (possibly negative) scaling, as claimed.
\end{proof} 

While the source distributions might not be uniquely identifiable for arbitrary co-variance matrices $\matr{\Sigma}$, they are indeed uniquely identifiable upto permutation and scaling for the noisy ICA model with isotropic Gaussian noise co-variance. For more general conditions that guarantee uniqueness of source distributions, please see \cite{Eriksson:2004, Eriksson:2006}.\\

\begin{cor}
If $\matr{\Sigma} = \sigma^2 \matr{I_p}$ and $\matr{\Sigma_1} = \sigma_1^2 \matr{I_p}$, then the source distributions are uniquely identifiable upto sign flips for $\matr{\Lambda^2} = \matr{I_q}$.
\end{cor}
\begin{proof}
Suppose $\matr{\Sigma} = \sigma^2 \matr{I_p}$ and $\matr{\Sigma_1} = \sigma_1^2 \matr{I_p}$. Then from \ref{eq9d} $\matr{\Sigma} = \matr{\Sigma_1}$ and thus $\mbox{diag}( \matr{W} \matr{\Sigma} \matr{W^T} ) = \mbox{diag} ( \matr{W} \matr{\Sigma_1} \matr{W^T} )$. The corollary then follows from Claim \ref{source_identifiability}.
\end{proof}

\begin{cor}
If $\matr{D} = \matr{D_1} = \matr{I_q}$, then the source distributions are uniquely identifiable up to sign flips for $\matr{\Lambda^2} = \matr{I_q}$.
\end{cor}
\begin{proof}
If $\matr{D} = \matr{D_1} = \matr{I_q}$, then noting that $\matr{P} \matr{P}^T = \matr{I_q}$, we get $\matr{D} = \matr{P} \matr{D_1} \matr{P^T}$. Hence from \ref{eq11d}, we get $\mbox{diag}( \matr{W} \matr{\Sigma} \matr{W^T} ) = \mbox{diag} ( \matr{W} \matr{\Sigma_1} \matr{W^T} )$. The corollary then follows from Claim \ref{source_identifiability}.
\end{proof}

\subsubsection{Why is there a run-to-run variability in estimated ICs?} \label{run-to-run}
From the discussion in section \ref{identifiability}, it is clear that for a noisy ICA model with isotropic diagonal additive Gaussian noise co-variance:
\begin{enumerate}
\item  The noisy ICA parameters $\mathcal{F} = \left\{ \vect{\mu}, \matr{A}, \matr{D}, \matr{\Sigma} \right\}$ are uniquely identifiable up to permutation and scaling.
\item The source distributions in $\vect{s}$ are uniquely identifiable upto permutation and scaling.
\end{enumerate}
While the above theoretical properties of ICA are reassuring, there are a number of practical difficulties that prevent the reliable estimation of ICs on real data: 
\begin{enumerate}
\item \textbf{Validity of the ICA model:} 

The assumption that the observed real data is generated by an ICA model is only that - an "assumption". If this assumption is not valid, then the uniqueness results do not hold anymore.

\item \textbf{Mutual information approximations:} 

From an information theoretic point of view, the ICA problem is solved by minimizing a contrast function which is an \textit{approximation} to the mutual information \citep{negentropy:1998} between the ICs that depends on the finite amount of observed data. Such an approximation is necessary, since we do not have access to the marginal source densities $p_{s_i}$. Different approximations to mutual information will lead to different objective functions and hence different solutions. This is one of the reasons why different ICA algorithms often produce different IC estimates even for the same data.

\item \textbf{Non-convexity of ICA objective functions:} 

The ICA contrast function is potentially non-convex and hence has multiple local minima. Since global minimization is a challenging problem by itself, most ICA algorithms will only converge to local minima of the ICA contrast function. The run-to-run variability of IC estimates will also depend on the number of local minima in a particular ICA contrast function.

\item \textbf{IC estimate corruption by Gaussian noise:} 

For noisy ICA, the IC realizations cannot be recovered exactly even if the true mixing matrix $\matr{A}$ and mean vector $\vect{\mu}$ are known in \ref{eq1d}. Commonly used estimators for recovering realization of ICs include the least squares \citep{Beckmann:2004} as well as the minimum mean square error (MMSE) \citep{Davies:2004}. Consider the least squares estimate $\vect{\hat{s}}$ of a realization of $\vect{s}$ based on $\vect{y}$:
\begin{align}\label{eq29d}
\vect{\hat{s}} = ( \matr{A^T} \matr{A} )^{-1} \matr{A^T} (\vect{y} - \vect{\mu}) = \vect{s} + ( \matr{A^T} \matr{A} )^{-1} \matr{A^T} \vect{\eta}
\end{align}
This means that even for known parameters, IC realization estimates $\vect{\hat{s}}$ will be corrupted by correlated Gaussian noise. Hence using different subsets of the data under the true model will also lead to variability in estimated ICs.

\item \textbf{Over-fitting of the ICA model:} 

Over specification of the model order leads to the problem of over-fitting in ICA. As we describe below, this can lead to (1) the phenomenon of IC "splitting" and (2) an increase in the variance of the IC estimates.

\textit{\textbf{1. IC "splitting"}}
 
Suppose that the true model order or the number of non-Gaussian sources in an ICA decomposition of $\vect{y}$ such as \ref{eq1d} is $q$. Then a fundamental result in \citep[Theorem 1]{Rao:1969} states that for any other ICA decomposition of $\vect{y}$, the number of non-Gaussian sources remains the same while the number of Gaussian sources can change. In other words, $\vect{y}$ cannot have two different ICA decompositions containing different number of non-Gaussian sources. 

In view of this fact, how can a model order $q$ ICA decomposition containing $q$ non-Gaussian sources be "split" into a $(q+1)$ ICA decomposition containing $(q+1)$ non-Gaussian sources when performing ICA estimation using an assumed model order of $(q+1)$? As we describe below, the order $(q+1)$ ICA decomposition is only an \textit{approximation} to the order $q$ ICA decomposition.

Let $\vect{a_i}$ be the $i$th column of $\matr{A}$ in \ref{eq1d}. In the presence of noise, it might be possible to \textit{approximate}: 
\begin{equation}\label{eq1gg}
\vect{a_i} s_i \approx \vect{a_i^1} s_i^1 + \vect{a_i^2} s_i^2
\end{equation}
Here:
\begin{itemize}
\item $\vect{a_i} s_i$ is the contribution of the $i$th non-Gaussian source $s_i$ to the ICA model \ref{eq1d}.
\item $s_i^1$ and $s_i^2$ are independent non-Gaussian random variables that are also independent with respect to all non-Gaussian sources $s_j, j \neq i$ in \ref{eq1d}.
\item $\vect{a_i^1}$ and $\vect{a_i^2}$ are the basis time courses corresponding to $s_i^1$ and $s_i^2$ respectively.
\item The time courses $\vect{a_i^1}$ and $\vect{a_i^2}$ look similar to each other.
\end{itemize}
Note that if $\vect{a_i^1} = \vect{a_i^2}$, then \ref{eq1gg} can be made into an equality by choosing $s_i = s_i^1 + s_i^2$. By replacing $\vect{a_i} s_i$ in \ref{eq1d} using \ref{eq1gg}, we arrive at an \textit{approximate} model order $(q+1)$ decomposition of $\vect{y}$. In this decomposition, the component $s_i$ from a model order $q$ decomposition appears to be "split" into two sub-components: $s_i^1$ and $s_i^2$.

\textit{\textbf{2. Inflated variance of IC estimates}}

Overestimation of model order will lead to over-fitting of the mixing matrix $\matr{A}$. In other words, $\matr{A}$ could have several columns that are highly correlated with each other. This could happen as a result of IC "splitting" as discussed above. Now, for a given realization $\matr{s}$, the variance of $\vect{\hat{s}}$ is given by $\mbox{Var}(\vect{\hat{s}}) = \sigma^2 (\matr{A^T} \matr{A})^{-1}$ (for isotropic Gaussian co-variance). An increase in number of columns of $\matr{A}$ and the fact that many of them are highly correlated implies that the variability of IC estimates $\mbox{Var}(\vect{\hat{s}})$ is inflated.

\end{enumerate}

In other words, running ICA multiple times on the same data or variations thereof with random initialization could produce different ICs.

\subsection{ICA algorithms, single subject ICA and group ICA}\label{ica_single_subject_and_group}

In this section, we give a brief summary of how the ICA parameters are estimated in practice and also summarize the two most common modes of ICA application to fMRI data - single subject ICA (section \ref{single_subject_ica}) and temporal concatenation based group ICA (section \ref{group_ica}).

Given several independent observations $\vect{y}$ as per the noisy ICA model \ref{eq1d}, most ICA algorithms estimate the ICA parameters $\mathcal{F} = \left\{ \vect{\mu}, \matr{A}, \matr{D}, \matr{\Sigma} \right\}$ and the realizations of $\vect{s}$ in 2 steps. We only consider the case with $\matr{\Sigma} = \sigma^2 \matr{I_p}$, since as shown in section \ref{identifiability}, the mixing matrix $\matr{A}$ and source distributions of $\vect{s}$ are identifiable upto permutation and scaling for this case.
\begin{enumerate}
\item First, the diagonal source co-variance is arbitrarily set as $\matr{D} = \matr{I_q}$. The mean vector $\vect{\mu}$ is estimated as $\mathbf{E}\left( \vect{y} \right)$.
Then, using PCA or PPCA \citep{Tipping:1999}, the mixing matrix $\matr{A}$ is estimated, upto an orthogonal rotation matrix $\matr{O}$, to be in a signal subspace which is spanned by the principal eigenvectors corresponding to the largest eigenvalues of the data co-variance matrix $\matr{E}\left[ (\vect{y} - \vect{\mu}) (\vect{y} - \vect{\mu})^T \right]$. The noise variance $\sigma^2$ is estimated in this step as well.
\item Next, an estimator $\vect{\hat{s}}$ for the source realizations is defined using techniques such as least squares or MMSE. The only unknown involved in these estimates is the orthogonal rotation matrix $\matr{O}$.
\item Finally, the non-Gaussianity of the empirical density of components of $\vect{\hat{s}}$ is optimized with respect to $\matr{O}$ using algorithms such as fixed point ICA \cite{negentropy:1998, FPICA:1999}.
\end{enumerate}
For more details on noisy ICA estimation, please see \citep{Beckmann:2004} and for more details on ICA algorithms, please see \citep{Hyvarinen:book}.
\subsubsection{Single subject ICA}\label{single_subject_ica}
How is ICA applied to single subject fMRI data?
Suppose we are given a single subject fMRI scan which we rearrange as a $p \times n$ 2D matrix $\matr{Y}$ in which column $i$ is the $p \times 1$ observed time-course $\vect{y_i}$ in the brain at voxel $i$. Observed time-courses $\vect{y_1}, \vect{y_2},\ldots, \vect{y_n}$ are considered to be $n$ independent realizations of $\vect{y}$ as per the linear ICA model \ref{eq1d}. Suppose $\matr{\hat{S}} = [ \vect{\hat{s}_1}, \vect{\hat{s}_2}, \ldots, \vect{\hat{s}_n}]$ is the $q \times n$ matrix containing the estimated source realizations at the $n$ voxels. The $j$th row of $\matr{\hat{S}}$ is the $j$th IC. In other words, we decompose the time by space fMRI 2D matrix into a set of basis time-courses and a set of $q$ 3D  IC maps using ICA. 

\subsubsection{Group ICA}\label{group_ica}
How is ICA applied to data from a group of subjects in fMRI?
Suppose we collect fMRI images from $m$ subjects. First, we register all subjects to a common space using a registration algorithm (e.g., affine registration). Next, we rearrange each of the fMRI scans into $m$ 2D matrices $\matr{Y_1} \ldots \matr{Y_m}$, each of size $p \times n$. Column $j$ in $\matr{Y_i}$ is the demeaned time-course observed at voxel location $j$ for subject $i$. The matrices $\matr{Y_1} \ldots \matr{Y_m}$ are temporally concatenated to get a $p m \times n$ matrix $\matr{Z}$ as follows:
\begin{equation}\label{eq1f}
\matr{Z} = \begin{pmatrix} \matr{Y_1} \\ \vdots \\ \matr{Y_i} \\ \vdots \\ \matr{Y_m} \end{pmatrix}
\end{equation}
Column $i$ of $\matr{Z}$ is the $p m \times 1$ vector $\vect{z_i}$ which is assumed to follow a linear ICA model \ref{eq1d}. $\vect{z_1}, \vect{z_2}, \ldots, \vect{z_n}$ are considered to be independent realizations of the model \ref{eq1d}. Suppose $\matr{\hat{S}_G} = [\vect{\hat{s}_1}, \vect{\hat{s}_2}, \ldots, \vect{\hat{s}_n}]$ is a $q \times n$ matrix containing the estimated source realizations at the $n$ voxels. The $j$th row of $\matr{\hat{S}_G}$ is the $j$th group IC. In group ICA, the joined time-series across subjects is modeled using noisy linear ICA. In practice, $\matr{Y_i}$ is the PCA reduced data set for subject $i$. The PCA reduction is either done separately for each subject using subject specific data co-variance \citep{Calhoun:2001} or an average data co-variance across subjects \citep{Beckmann:2005}. The average co-variance approach requires each subject to have the same number of time points in fMRI scans.

\subsection{The original RAICAR algorithm}\label{orig_raicar}

In this section, we give a brief introduction to the RAICAR algorithm of \citep{Yang:2008}.  Suppose we are given a data set which we decompose into $n_C$ ICs using ICA (e.g., single subject or group ICA). Our goal is to assess which ICs consistently show up in multiple ICA runs i.e., the reproducibility of each of these $n_C$ ICs. To that extent, we run the ICA algorithm $K$ times. Suppose $\vect{x}_j^{(m)}$ is the $n \times 1$ vector (e.g. spatial ICA map re-arranged into a vector) of the $j$th IC from $m$th ICA run. Suppose $\matr{G}_{lm}$ is a $n_C \times n_C$ absolute spatial cross-correlation coefficient matrix between the ICs from runs $l$ and $m$:

\begin{equation}\label{eq1a}
\matr{G}_{lm}(i,j) = | \mbox{corrcoef}(\vect{x}_i^{(l)}, \vect{x}_j^{(m)}) |
\end{equation}
where $| . |$ denotes absolute value. $\matr{G}_{lm}(i,j)$ is the absolute spatial cross-correlation coefficient between IC $i$ from run $l$ and IC $j$ from run $m$. The matrices $\matr{G}_{lm}$ are then arranged as elements of a $K \times K$ \textit{block-matrix} $\matr{G}$ such that the $l$th row and $m$th column of $\matr{G}$ is $\matr{G}(l,m) = \matr{G}_{lm}$ (see Figure \ref{figure2}). This block matrix $\matr{G}$ is the starting point for a RAICAR across-run component matching process. 

Since ICs within a particular run cannot be matched to each other, the $n_C \times n_C$ matrices $\matr{G}(l,l), l = 1 \ldots K$ along the block-diagonal of $\matr{G}$ are set to $\matr{0}$ as shown in Figure \ref{figure2} with a gray color. The following steps are involved in a RAICAR analysis:

\begin{enumerate}
\item Find the maximal element of $\matr{G}$. Suppose this maximum occurs in matrix $\matr{G}_{lm}$ at position $(i,j)$. Hence component $i$ from run $l$ matches component $j$ from run $m$. Let us label this matched component by $MC_1$ (the first matched component).

\item Next, we attempt to find from each run $s$ ($s \neq l$ and $s \neq m$) a component that matches with component $MC_1$. Suppose element $(a_s, j)$ is the maximal element in the $j$th column of $\matr{G}_{sm}$. Then component $a_s$ is the best matching component from run $s$ with the $j$th component from run $m$.

Similarly, suppose element $(i, b_s)$ is the maximal element in the $i$th row of $\matr{G}_{ls}$. Then component $b_s$ is the best matching component from run $s$ with component $i$ from run $l$. As noted in \citep{Yang:2008}, in most cases $a_s = b_s$. However, it is possible that $a_s \neq b_s$. Hence the component number $e_s$ matching $MC_1$ from run $s$ is defined as follows:

\begin{equation}\label{eq2a}
e_s =\begin{cases}
a_s & \text{if $\matr{G}_{sm}(a_s, j) \ge \matr{G}_{ls}(i, b_s)$},\\
b_s & \text{if $\matr{G}_{sm}(a_s, j) < \matr{G}_{ls}(i, b_s)$}.
\end{cases}
\end{equation}

We would also like to remove component $e_s$ of run $s$ from further consideration during the matching process. To that extent, we zero out the $e_s$th row from $\matr{G}_{sr}, r = 1 \ldots K$ and the $e_s$th column from $\matr{G}_{rs}, r = 1 \ldots K$.

\item Once a matching component $e_s$ has been found for all runs $s \neq l,m$, we also zero out the $i$th row from $\matr{G}_{lr}, r = 1 \ldots K$ and the ith column from $\matr{G}_{rl}, r = 1 \ldots K$. Similarly, we zero out the $j$th column from $\matr{G}_{rm}, r = 1 \ldots K$ and the $j$th row from $\matr{G}_{mr}, r = 1 \ldots K$. This eliminates component $i$ from run $l$ and component $j$ from run $m$ from further consideration during the matching process.

\item Steps 1-3 complete the matching process for one IC component across runs. These steps are repeated until $n_C$ components are matched across the $K$ runs. We label the matched component $s$ as $MC_s$ which contains a set of $K$ matching ICs one from each of the $K$ ICA runs.

\end{enumerate}

Suppose matched component $s$, $MC_s$ consists of the matched ICs $\vect{x}_{i_1}^{(1)}, \vect{x}_{i_2}^{(2)},\ldots,\vect{x}_{i_K}^{(K)}$. Form the $K \times K$ cross-correlation matrix $H_{MC_s}$ between the matched components in $MC_s$. The $(a,b)$th element of this matrix is simply:
\begin{equation}\label{eq3a}
H_{MC_s}(a,b) = | \mbox{corrcoef}\left( \vect{x}_{i_a}^{(a)}, \vect{x}_{i_b}^{(b)} \right) |
\end{equation}
The normalized reproducibility of $MC_s$ is then defined as:

\begin{equation}\label{eq4a}
\mbox{Reproducibility}(MC_s) = \left( \frac{2}{(K-1)K} \right) \sum_{a=1}^K \sum_{b=a+1}^K H_{MC_s}(a,b)
\end{equation}

The double sum in \ref{eq4a} is simply the sum of the upper triangular part of $H_{MC_s}$ excluding the diagonal. The normalizing factor $\frac{ (K-1)K }{2}$ is simply the maximum possible value of this sum. Hence the normalized reproducibility satisfies: $\mbox{Reproducibility}(MC_s) \le 1$.

Note that our definition of normalized reproducibility is slightly different from that in \cite{Yang:2008}. Whereas \cite{Yang:2008} averages the \textit{thresholded} absolute correlation coefficients, we simply average the \textit{un-thresholded} absolute correlation coefficients to compute reproducibility thereby avoiding the selection of a threshold on the absolute correlation coefficients.

\begin{figure}[htbp]
\begin{center}
\includegraphics[width=6.2in]{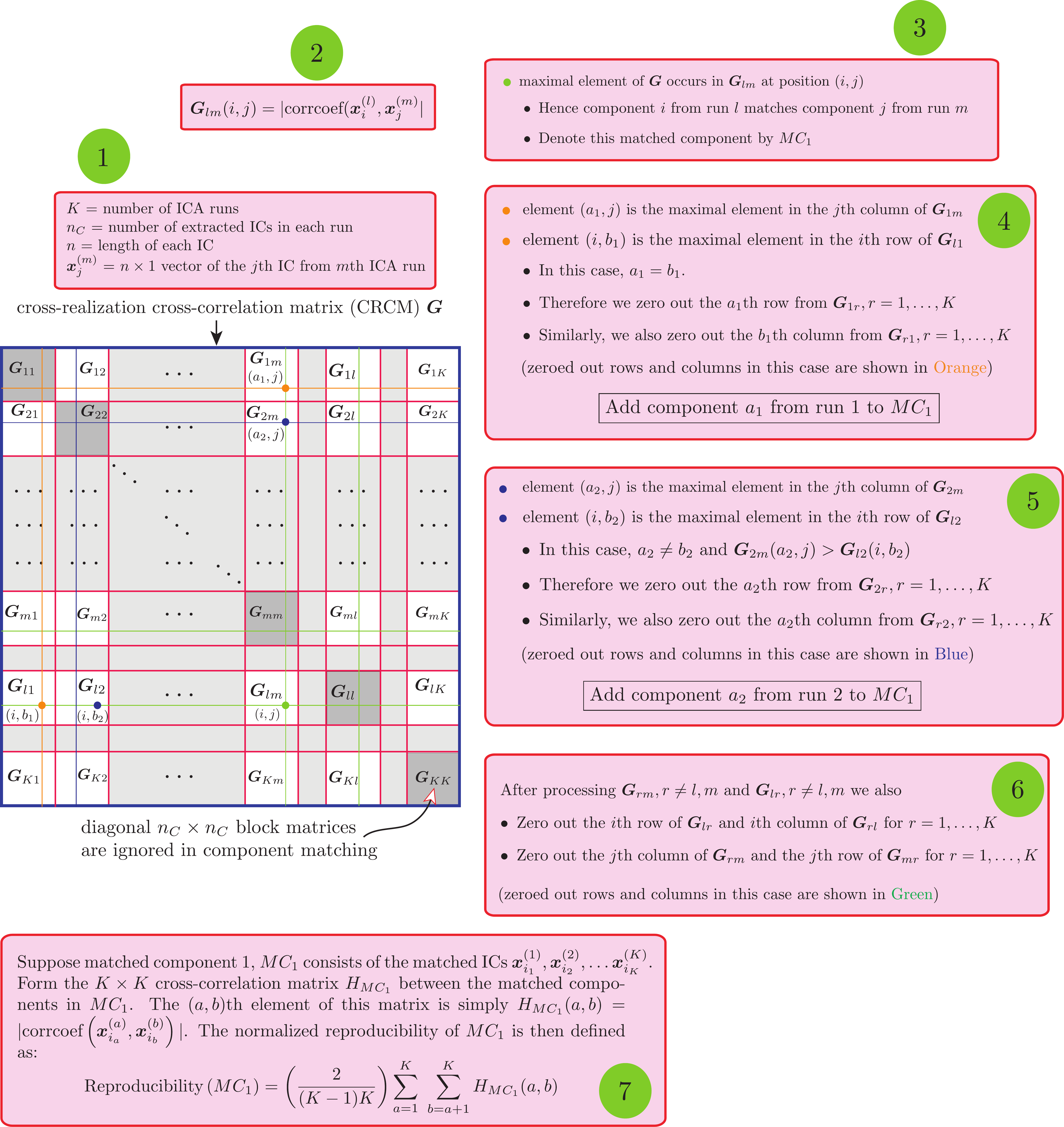}
\caption{Pictorial depiction of the original RAICAR algorithm \citep{Yang:2008}. The ICA algorithm is run $K$ times with each run producing $n_C$ ICs. $\matr{G}$ is a $K \times K$ block matrix with elements $\matr{G}(l,m) = \matr{G}_{lm}$ where $\matr{G}_{lm}$ is the $n_C \times n_C$ absolute spatial cross-correlation matrix between ICs from runs $l$ and $m$. The numbered green circles indicate the sequence of steps in applying RAICAR to a given data set. Our definition of normalized reproducibility in box 7 averages un-thresholded correlation coefficients thereby avoiding the selection of a correlation coefficient threshold prior to averaging.}
\label{figure2}
\end{center}
\end{figure}

\subsection{The RAICAR-N enhancement}\label{raicar_n}

In this section, we describe how to compute reproducibility $p$-values for each matched component in RAICAR.
Note that the RAICAR "component matching" process can be used to assess the reproducibility of \textit{any} spatial component maps - not necessarily ICA maps. For instance, RAICAR can be used to assess the reproducibility of a set of PCA maps across subjects. 

In order to generate reproducibility $p$-values for the matched component maps:
\begin{enumerate}
\item We need to determine the distribution of normalized reproducibility that we get from the RAICAR "component matching" process when the input to RAICAR represents a set of "non-reproducible component maps" across the $K$ runs. 
\item In addition, we would also like to preserve the \textit{overall structure} seen in the \textit{observed} sets of spatial component maps across the $K$ runs when generating sets of "non-reproducible component maps" across the $K$ runs.
\end{enumerate}
Hence for IC reproducibility assessment, we propose to use the original set of ICs across the $K$ runs to generate the "non-reproducible component maps" across the $K$ runs. 

Suppose $K$ ICA runs are submitted to RAICAR which gives us a $n_C \times 1$ vector of \textit{observed} normalized reproducibility values $\mbox{Reproducibility}(MC_i), i = 1 \ldots n_C$ - one for each IC. We propose to attach $p$-values for measuring the reproducibility of each IC in a data-driven fashion as follows:

\begin{enumerate}
\item First, we label the $K n_C$ ICs across the $K$ runs using unique integers. In run 1, the ICs are labelled using integers $1,\ldots,n_C$. In run 2, the ICs are labelled using integers $(n_C + 1),\ldots,2 n_C$ and so on. In run $K$, the ICs are labelled using integers $(K-1) n_C + 1,\ldots,K n_C$.

\item Our "null" hypothesis is: 
\begin{align}\label{eq1b}
\mathbf{H_0}: \,\,\, &\mbox{\textbf{None of the ICs are reproducible}} \\ \nonumber
		       &\mbox{Hence, we can randomly label component $i$ from run $l$ as component $d$ from run $s$ }
\end{align}

To do this, we randomly permute the integers $1,2,\ldots,K n_C$ to get the permuted integers $p(1),p(2),\ldots,p(K n_C)$. Obviously $p(i) \neq p(j) \mbox{ if } i \neq j$.

\item The $K$ sets "non-reproducible component runs under $\mathbf{H_0}$" are constructed by assigning components with labels:
	\begin{itemize}
		\item $p(1),\ldots,p(n_C)$ to run 1 under $\mathbf{H_0}$.
		\item $p(n_C+1),\ldots,p(2 n_C)$ to run 2 under $\mathbf{H_0}$
		\item $p\left( (K-1) n_C + 1 \right),\ldots,p(K n_C)$ to run $K$ under $\mathbf{H_0}$
	\end{itemize}

\item After $K$ runs have been generated under $\mathbf{H_0}$, we subject these to a RAICAR analysis. This gives us $n_C$ values of normalized reproducibility, one for each matched component under $\mathbf{H_0}$.

\item Steps 1-4 are repeated $R$ times to build up a \textit{pooled} $R n_C \times 1$ vector of normalized reproducibility $\mbox{\textbf{Reproducibility}}_{Null}$ under $\mathbf{H_0}$.

\item Finally, we assign a $p$-value for reproducibility to each matched IC across the $K$ runs. The observed reproducibility for $i$th matched IC is $\mbox{Reproducibility}(MC_i)$ and its $p$-value is:
\begin{equation}\label{reproducibility}
\mbox{Reproducibility}_{pval}(MC_i) = \frac{ \left\{ \mbox{no. of } \mbox{\textbf{Reproducibility}}_{Null} \ge \mbox{Reproducibility}(MC_i) \right\} + 1}{R n_C + 1}
\end{equation}

\item Only those components with $\mbox{Reproducibility}_{pval}(MC_i) < p_{crit}$ are considered to be significantly reproducible. We can use a fixed and objective value for $p_{crit}$ such as $0.05$. Note that this fixed cutoff is independent of the amount of variability in the input to RAICAR-N. Please see Figure \ref{figure3} for a pictorial depiction of this process.
\end{enumerate}

\begin{figure}[htbp]
\begin{center}
\includegraphics[width=4.1in]{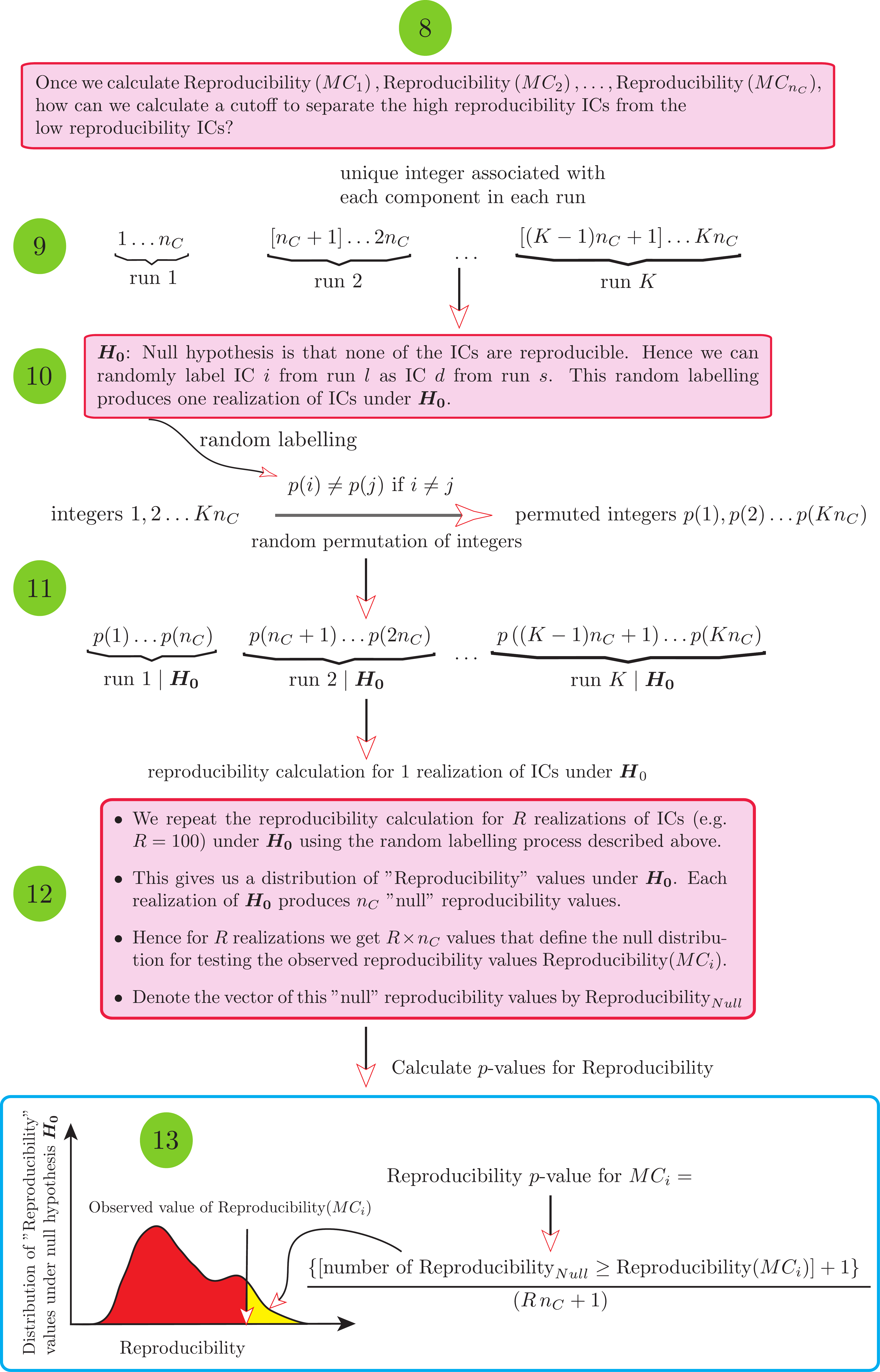}
\caption{Pictorial depiction of the process for generating a "null" distribution in RAICAR-N. Our "null" hypothesis is: "$\mathbf{H_0}$: None of the ICs are reproducible. Hence, we can randomly label IC $i$ from run $l$ as IC $d$ from run $s$". Therefore we randomly split the $K n_C$ ICs across $K$ runs into $K$ parts and run the RAICAR algorithm on each set of randomly split ICs. This gives us a set of "null" reproducibility values which can be used to compute $p$-values for the observed reproducibility of ICs in the original RAICAR run. The green circles indicate the sequence of steps for generating the "null" distribution after the steps in Figure \ref{figure2}.}
\label{figure3}
\end{center}
\end{figure}

\subsection{How many subjects should be used per group ICA run in RAICAR-N?}\label{choosingL}

The input to RAICAR-N can either be single subject ICA runs or group ICA runs across a set of subjects. Note that the individual subject ICA runs are spatially unconstrained whereas a group ICA spatially constrains the group ICs across a set of subjects. Hence the number of ICs that can be declared as significantly reproducible at the group level are usually more than those that can be declared significantly reproducible at the single subject level. Hence the following question is relevant:

\begin{figure}[htbp]
\begin{center}
\includegraphics[width=6.3in]{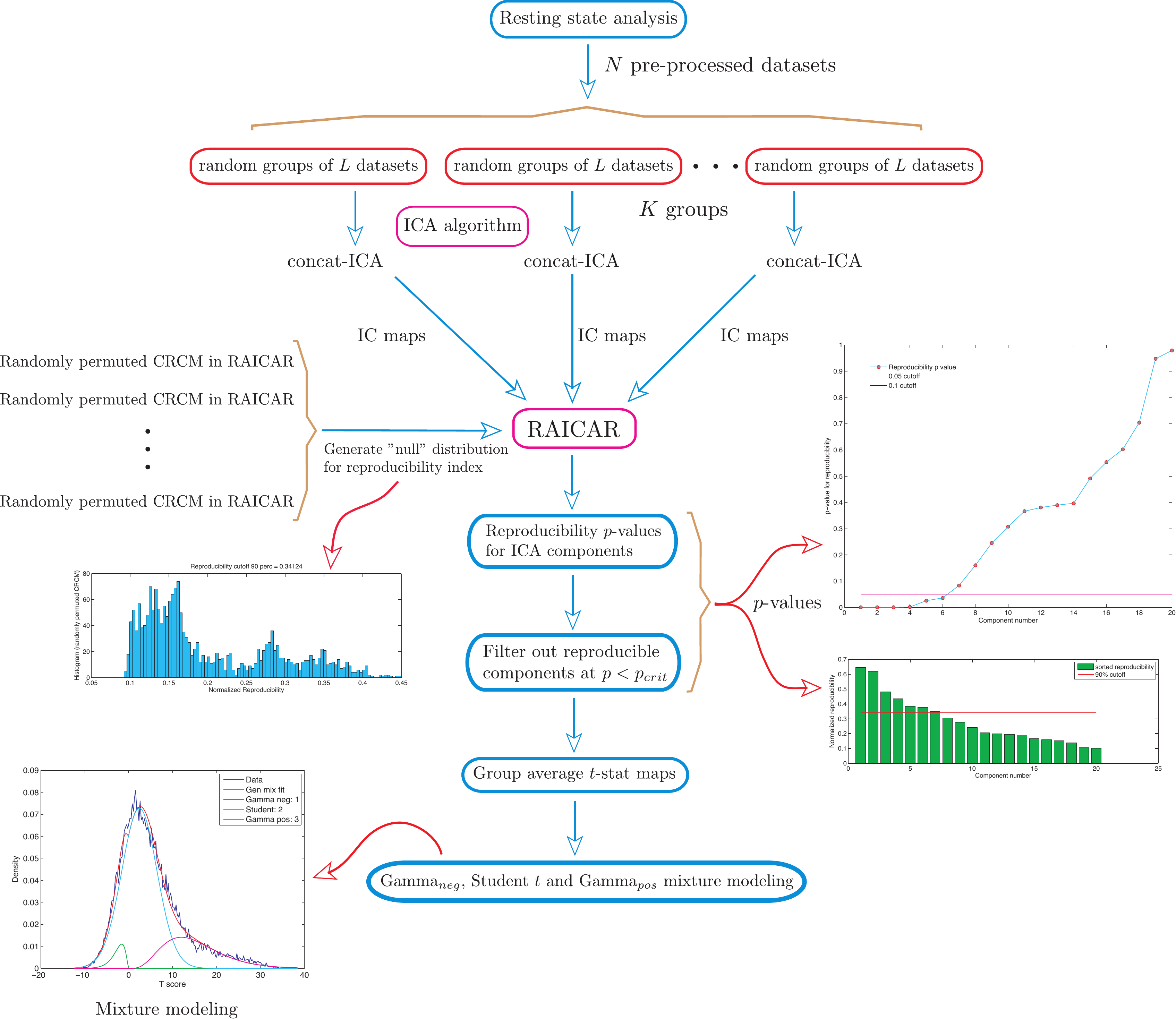}
\caption{Flowchart for a group ICA based RAICAR-N analysis. The $N$ single subject data sets are first pre-processed and subsequently bootstrapped to create $K$ groups, each group containing $L$ distinct subjects. Each group of $L$ subjects is submitted to a temporal concatenation group ICA analysis. The resulting IC maps (either raw ICs or ICs scaled by noise standard deviation) are subjected to a RAICAR analysis. The cross-realization cross correlation matrix (CRCM) is randomly permuted multiple times: $\matr{G} \rightarrow \matr{G}(\vect{g},\vect{g})$ where $\vect{g}$ is a random permutation of integers from $1,\ldots,K n_C$. The permuted CRCMs are subjected to a RAICAR analysis to generate a realization of reproducibility values under the "null" hypothesis. The computed "null" distribution of reproducibility values is used to assign $p$ values to the observed reproducibility of the original RAICAR run. Finally, reproducible ICs are averaged using a random effects analysis and the resulting $t$-statistic images are subjected to Gamma$_{neg}$, Student $t$ and Gamma$_{pos}$ mixture modeling.}
\label{figure4}
\end{center}
\end{figure}

Suppose we have a group of $N$ subjects. We randomly select $L$ subjects and form a single group of subjects. We repeat this process $K$ times to get $K$ groups of $L$ subjects each of which is subjected to a group ICA analysis. Given the number of subjects $N$, how should we choose $L$ and $K$?

First, we discuss the choice of $L$. If $L=N$ then each of the $K$ groups will contain the same $N$ subjects and hence there will be no diversity in the $K$ groups. We would like to control the amount of diversity in the $K$ groups of $L$ subjects. Consider any 2 subjects $X$ and $Y$. The probability $P_{XY}(L)$ that both $X$ and $Y$ appear in a set of $L$ randomly chosen subjects from $N$ subjects is given by:

\begin{equation}\label{eq1}
P_{XY}(L) = \frac{ {N-2 \choose L-2} }{ {N \choose L} }
\end{equation}

The expected number of times that $X$ and $Y$ appear together in sets of $L$ subjects out of $K$ independently drawn sets is:
\begin{equation}\label{eq2}
E_{XY}(L) = K \, P_{XY}(L)
\end{equation}

Ideally, we would like $E_{XY}(L)$ to be only a small fraction of $K$. Hence we impose the restriction:

\begin{equation}\label{eq3}
E_{XY}(L) = K \, P_{XY}(L) \le \alpha_{max} \, K
\end{equation}

where $\alpha_{max}$ is a user defined constant such as $\alpha_{max} = 0.05$. This implies that the chosen value of $L$ must satisfy:

\begin{equation}\label{eq4}
P_{XY}(L) \le \alpha_{max}
\end{equation}

In practice, we choose the largest value of $L$ that satisfies this inequality. As shown in Figure \ref{figure5}, if $N = 23$ and $\alpha_{max} = 0.05$ then the largest value of $L$ that satisfies \ref{eq4} is $L = 5$. 
\begin{figure}[htbp]
\begin{center}
\includegraphics[width=3.0in]{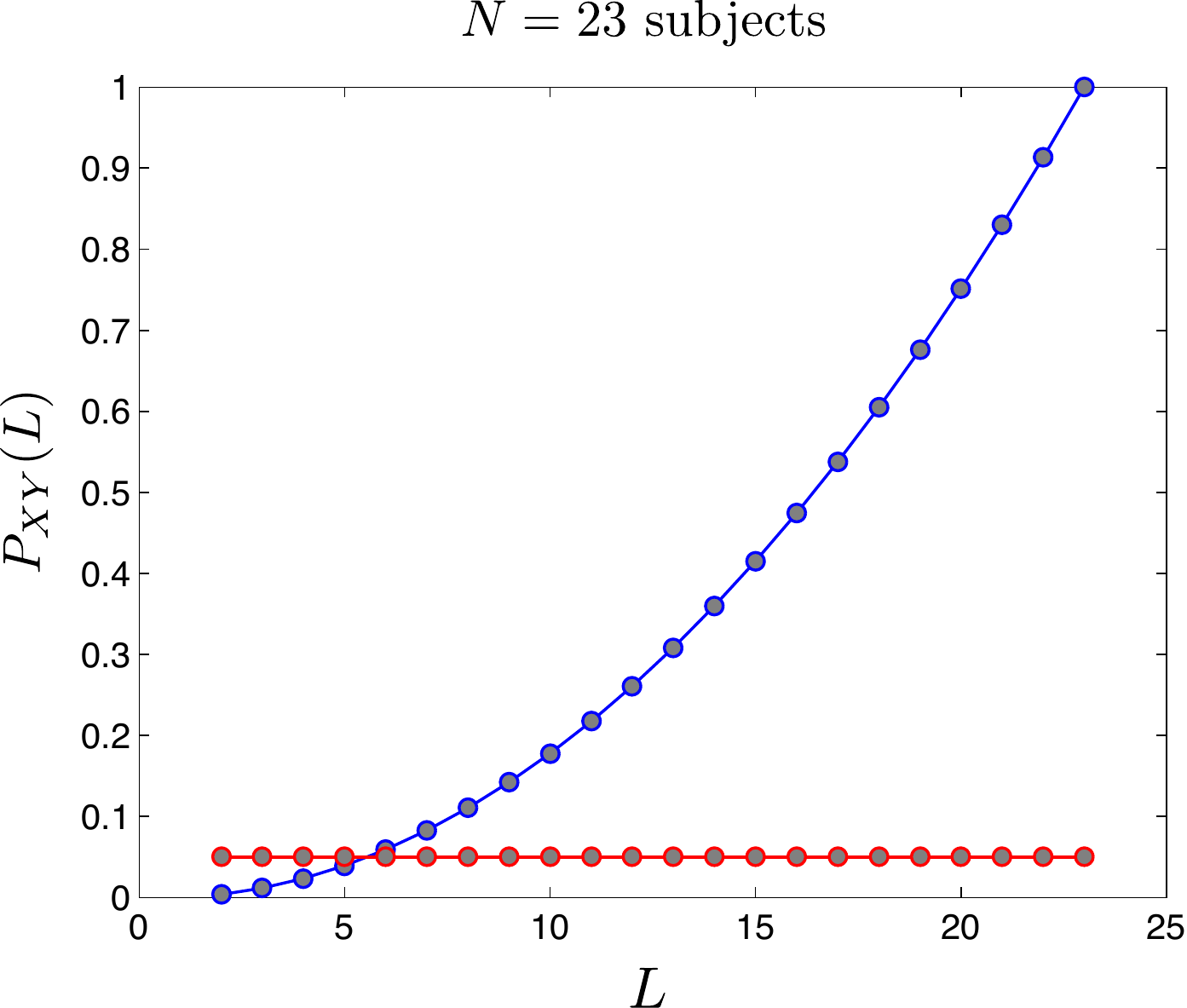}
\caption{Figure shows a plot of $P_{XY}(L)$ vs $L$ for $N = 23$ in \textcolor{blue}{blue}. The \textcolor{red}{red} line shows the $\alpha_{max} = 0.05$ cutoff. The largest value of $L$ for which $P_{XY}(L) \le 0.05$ is $L = 5$.}
\label{figure5}
\end{center}
\end{figure}

The number of group ICA runs $K$ should be as large as possible. From our experiments on real fMRI data we can roughly say that values of $K > 50$ give equivalent results.

\subsection{How to display the estimated non-Gaussian spatial structure in ICA maps?}\label{non_gaussian_structure}

The ICs have been optimized for non-Gaussianity. However, there can be many types of non-Gaussian distributions. It has been empirically found that the non-Gaussian distributions of ICs found in fMRI data have the following structure: 
\begin{enumerate}
\item A central Gaussian looking part and 
\item A tail that extends out on either end of the Gaussian
\end{enumerate} 
It has been suggested in \citep{Beckmann:2004} that a Gaussian/Gamma mixture model can be fitted to this distribution and the Gamma components can be thought of as representatives of the non-Gaussian structure. We follow a similar approach:

\begin{enumerate}

\item The output of a RAICAR-N analysis is a set of spatial ICA maps (either $z$-transformed maps or raw maps) concatenated into a 4-D volume. 

\item We do a voxelwise transformation to Normality using the voxelwise empirical cumulative distribution function as described in \citep{vanAlbada:2007}. 

\item Next, we submit the resulting 4-D volume to a voxelwise group analysis using ordinary least squares. The design matrix for group analysis depends on the question being considered. In our case, the design matrix was simply a single group average design. 

\item The resulting $t$-statistic maps are subjected to Student $t$, Gamma$_{pos}$ and Gamma$_{neg}$ mixture modeling. The logic is that if the original ICA maps are pure Gaussian (i.e., have no interesting non-Gaussian structure) then the result of a group average analysis will be a pure Student $t$ map which will be captured by a single Student $t$ (i.e., the Gamma$_{pos}$ and Gamma$_{neg}$ will be driven to $0$ class fractions). Hence the "null" hypothesis will be correctly accounted for. 

\item If the Gamma distributions have $> 0.5$ posterior probability at some voxels then those voxels are displayed in color to indicate the presence of significant non-Gaussian structure over and above the background Student $t$ distribution.

\end{enumerate}
Examples of Student $t$, Gamma$_{pos}$ and Gamma$_{neg}$ mixture model fits are shown in Figure \ref{figure6}.

\begin{figure}[htbp]
\begin{center}
\includegraphics[width=6.0in]{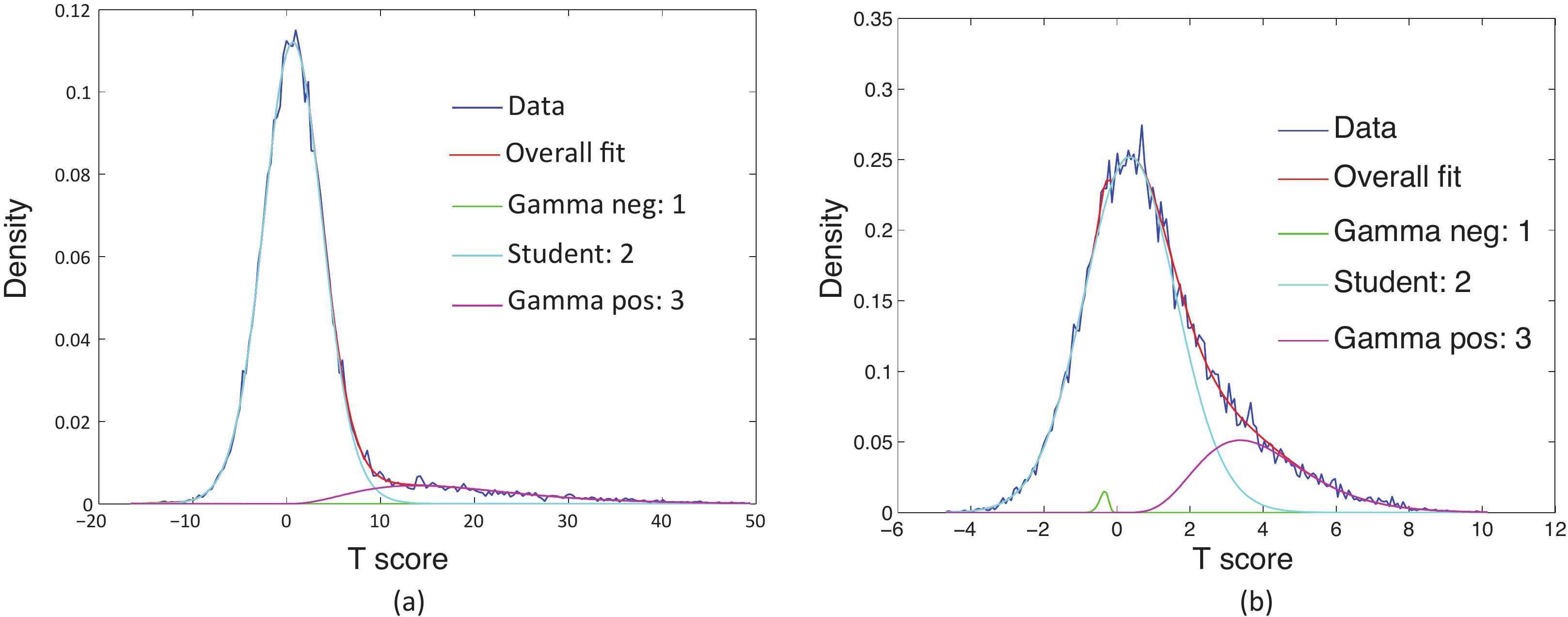}
\caption{Examples of displaying non-Gaussian spatial structure using a Student $t$, Gamma$_{pos}$ and Gamma$_{neg}$ mixture model. Notice how the Gamma$_{neg}$ density is driven to near $0$ class fraction in the absence of significant negative non-Gaussian structure.}
\label{figure6}
\end{center}
\end{figure}

\begin{figure}[htbp]
\begin{center}
\includegraphics[width=6.0in]{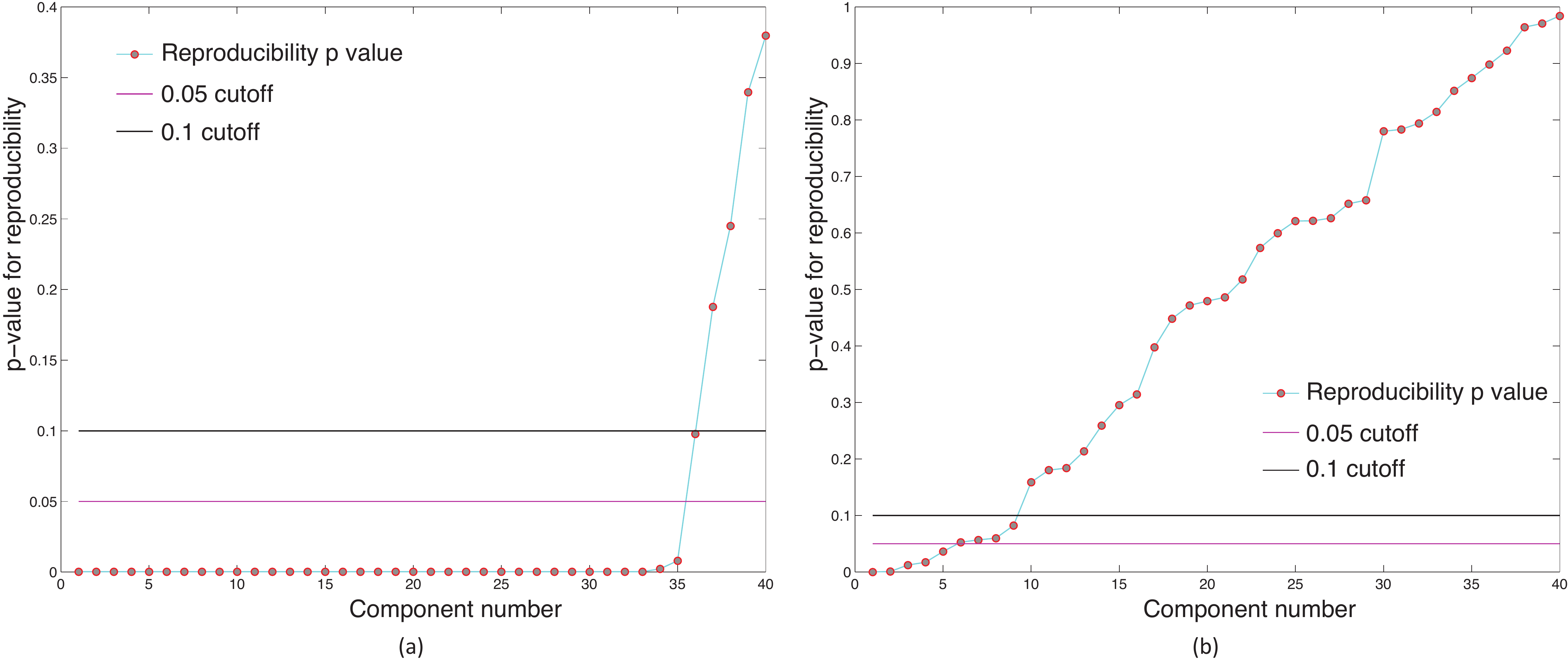}
\caption{$p$-value cutoffs for within and across single subject analysis using RAICAR-N. This figure illustrates the intuitive fact that within subject ICA runs are much more reproducible compared to across subject ICA runs.}
\label{figure7}
\end{center}
\end{figure}

\section{Experiments and Results}

\subsection{Human rsfMRI data}
rsfMRI data titled: \verb+Baltimore (Pekar, J.J./Mostofsky, S.H.; n = 23 [8M/15F]; ages: 20-40;+\\
\verb+ TR = 2.5; # slices = 47; # timepoints = 123)+, a part of the 1000 functional connectomes project, was downloaded from the Neuroimaging Informatics Tools and Resources Clearinghouse (NITRC): \url{http://www.nitrc.org/projects/fcon_1000/}. 

\subsection{Preprocessing}
Data was analyzed using tools from the FMRIB software library (FSL: \url{http://www.fmrib.ox.ac.uk/fsl/}). Preprocessing steps included motion correction, brain extraction, spatial smoothing with an isotropic Gaussian kernel of 5mm FWHM and 100s high-pass temporal filtering. Spatial ICA was performed using a noisy ICA model as implemented in FSL MELODIC \citep{Beckmann:2004} in either single subject or multi-subject temporal concatenation mode also called group ICA. Please see section \ref{ica_single_subject_and_group} for a brief summary of single subject ICA and group ICA. In each case, we fixed the model order of ICA at $q=40$ to be consistent with the model order range typically extracted in rsfMRI and fMRI \citep{Smith:2009, Esposito:2005}. For temporal concatenation based group ICA, single subject data was first affinely registered to the MNI 152 brain and subsequently resampled to 4x4x4 resolution (MNI 4x4x4) to decrease computational load. 

\subsection{RAICAR-N analysis with 1 ICA run per subject}

Spatial ICA was run once for each of the $N = 23$ subjects in their native space. The resulting set of ICA components across subjects were transformed to MNI 4x4x4 space and were submitted to a RAICAR-N analysis.\footnote[1]{In all RAICAR-N analyses reported in this article, we used the \mbox{$z$}-transformed IC maps - which are basically the raw IC maps divided by a voxelwise estimate of noise standard deviation (named as \texttt{melodic\_IC.nii.gz} in MELODIC). It is also possible to use the raw IC maps as inputs to RAICAR-N.} ICA components were sorted according to their reproducibility and $p$-values were computed for each ICA component. Please see Figure \ref{figure_pval_23subj}. 

\begin{figure}[htbp]
\begin{center}
\includegraphics[width=6.0in]{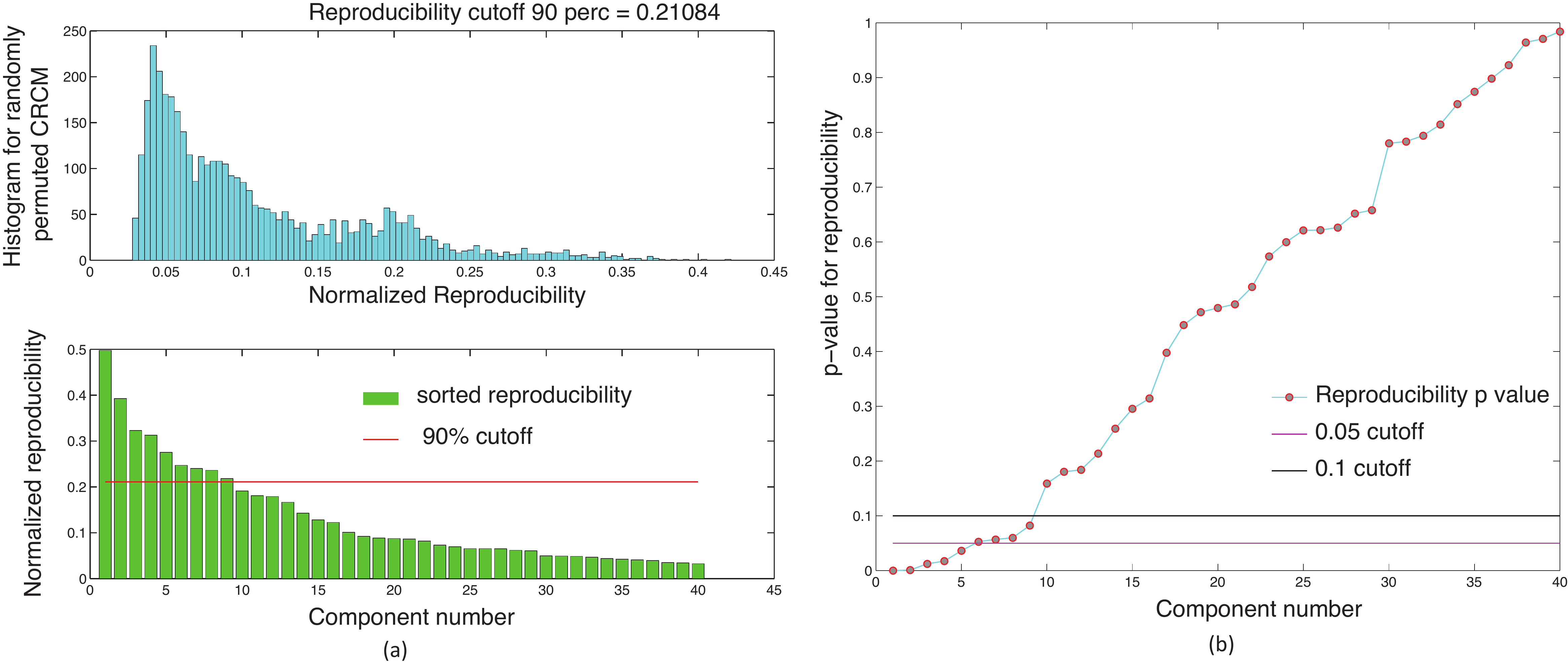}
\caption{Single subject rsfMRI ICA runs across 23 subjects were combined using a RAICAR-N analysis. Figure (a) shows the observed values of normalized reproducibility (bottom) as well as the "null" distribution of normalized reproducibility across $R=100$ simulations (top). Figure (b) shows the $p$-values for each IC along with the $0.05$ and $0.1$ cutoff lines.}
\label{figure_pval_23subj}
\end{center}
\end{figure}

We compared the reproducible RSNs from the single subject RAICAR-N analysis to the group RSN maps reported in literature \citep{Beckmann:2005}. Please see Figure \ref{figure_top8maps_summary}.

\begin{figure}[htbp]
\begin{center}
\includegraphics[width=5.8in]{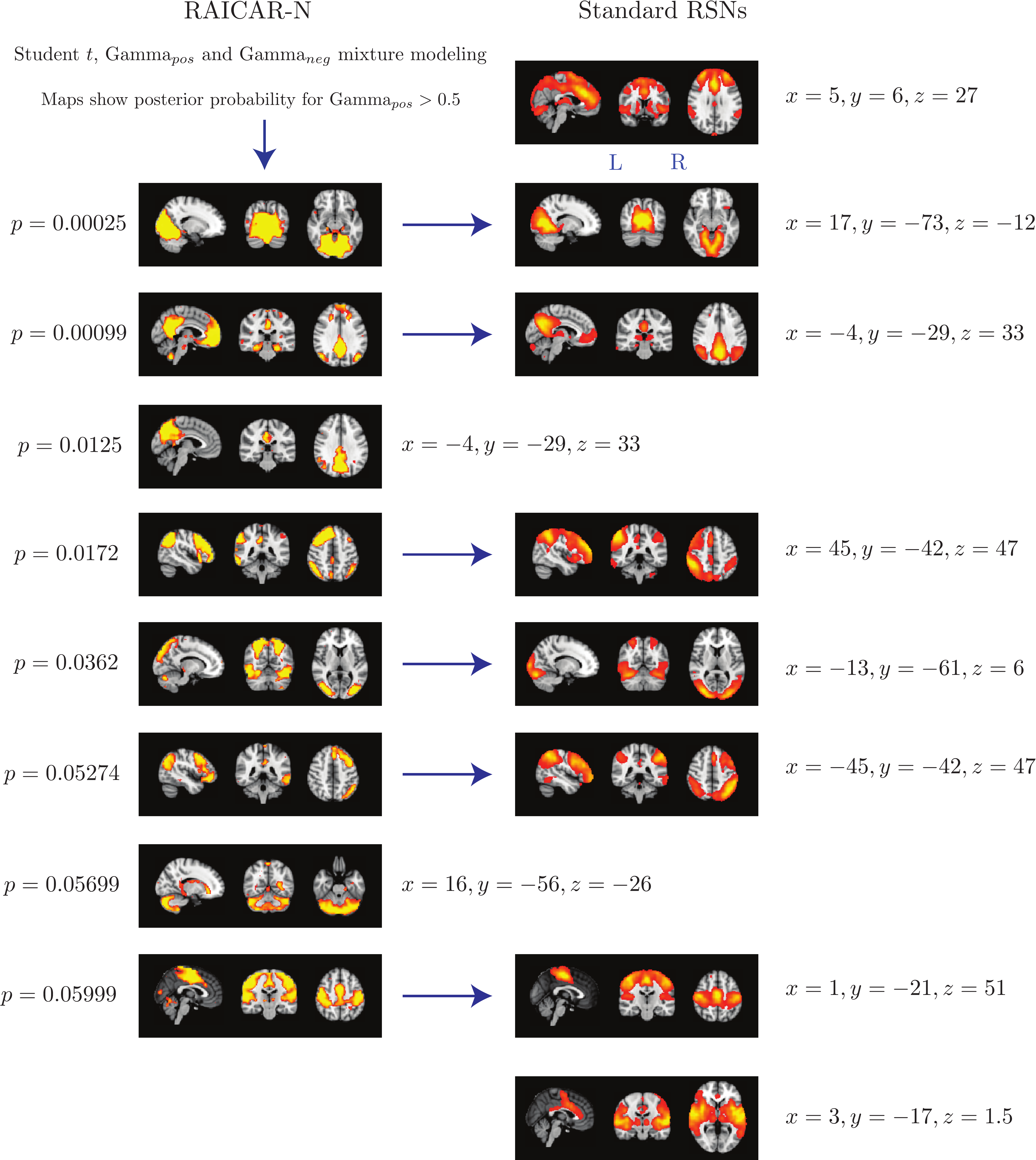}
\caption{The top 8 "reproducible" ICs from a RAICAR-N analysis on single subject ICA runs compared with standard RSN maps reported in literature \citep{Beckmann:2005}. We are able to declare 4 "standard" RSNs as significantly reproducible at a $p$-value $<0.05$. There are 2 other "standard" RSNs that achieve a reproducibility $p$-value between $0.05$ and $0.06$ as well as 2 "non-standard" RSNs that achieve $p$-values of $0.0125$ and $0.05699$ respectively. We also could not find 2 of the published RSNs in \citep{Beckmann:2005} as reproducible in single subject ICA runs.}
\label{figure_top8maps_summary}
\end{center}
\end{figure}

To summarize, when single subject ICA runs are combined across subjects: 
\begin{itemize}
\item We are able to declare 4 "standard" RSNs as significantly reproducible at a $p$-value $< 0.05$. 
\item There are 2 other "standard" RSNs that achieve a reproducibility $p$-value between 0.05 and 0.06. 
\item There are 2 other "non-standard" RSNs that are of interest: one achieves a $p$-value of 0.0125 and the other achieves a $p$-value of 0.05699. 
\end{itemize}

\subsection{RAICAR-N on random sets of 5 subjects - 50 group ICA runs}
To promote diversity across the group ICA runs, as discussed in section \ref{choosingL}, $L = 5$ subjects were drawn at random from the group of $N = 23$ subjects and submitted to a temporal concatenation based group ICA. This process was repeated $K = 50$ times and the resulting set of 50 group ICA maps were submitted to a RAICAR-N analysis. ICA components were sorted according to their reproducibility and $p$-values were computed for each ICA component. Please see Figure \ref{figure_pval_50groupof5}.

\begin{figure}[htbp]
\begin{center}
\includegraphics[width=6.0in]{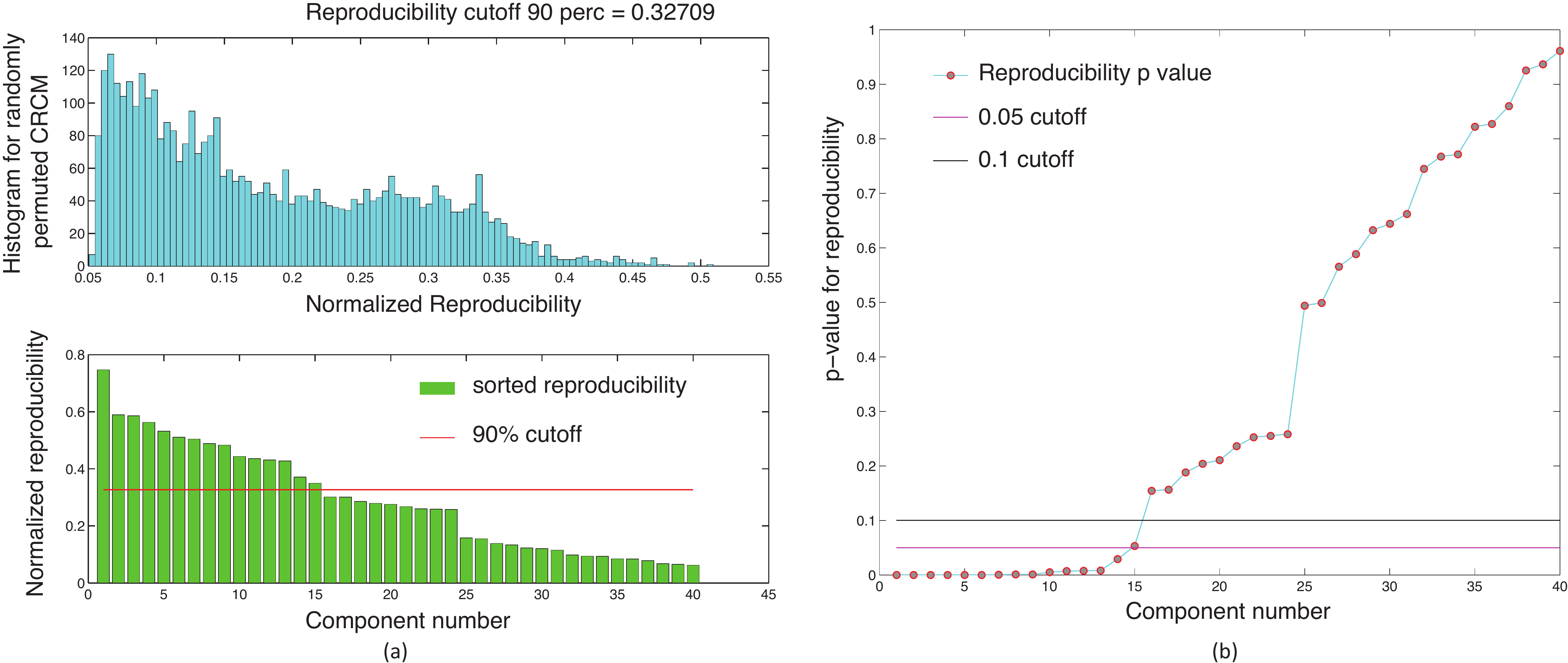}
\caption{$L=5$ subjects were randomly drawn from the set of $N=23$ subjects and submitted to a temporal concatenation based group ICA. This process was repeated $K=50$ times and the resulting ICA maps were submitted to a RAICAR-N analysis. Figure (a) shows the observed values of normalized reproducibility (bottom) as well as the "null" distribution of normalized reproducibility across $R=100$ simulations (top). Figure (b) shows the $p$-values for each IC along with the $0.05$ and $0.1$ cutoff lines.}
\label{figure_pval_50groupof5}
\end{center}
\end{figure}

We compared the reproducible RSNs from the single subject RAICAR-N analysis to the RSN maps reported in literature \citep{Beckmann:2005}. Please see Figure \ref{figure_top15maps_summary}.

\begin{figure}[htbp]
\begin{center}
\includegraphics[width=4.7in]{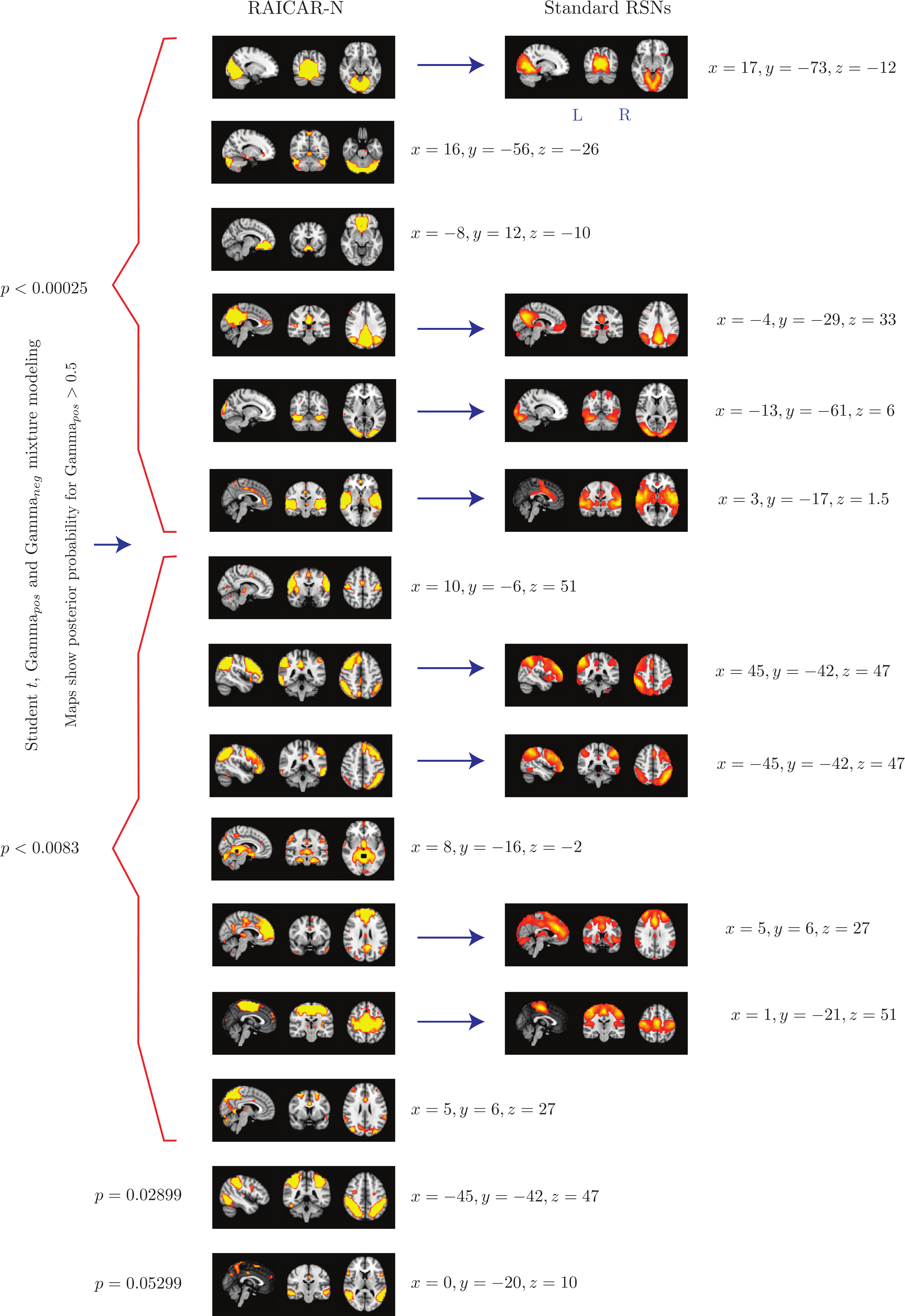}
\caption{The top 15 "reproducible" ICs from $K=50$ runs of $L=5$ subject group ICA RAICAR-N analysis compared with standard RSN maps reported in literature \citep{Beckmann:2005}. We are able to declare 8 "standard" RSNs as significantly reproducible at a $p$-value of $<0.05$. There are 6 other "non-standard" RSNs that can be declared as significantly reproducible at a $p$-value of $<0.05$ and 1 other "non-standard" RSN that achieves a $p$-value of $0.05299$.}
\label{figure_top15maps_summary}
\end{center}
\end{figure}

In summary, when 50 random 5 subject group ICA runs (from a population of 23 subjects) are combined using RAICAR-N: 
\begin{itemize}
\item We are able to declare 8 "standard" RSNs as significantly reproducible at a $p$-value $< 0.05$. 
\item There are 6 other "non-standard" RSNs that can be declared as significantly reproducible at a $p$-value $< 0.05$. 
\item There is 1 other "non-standard" RSN that achieves a $p$-value of 0.05299. 
\end{itemize}

\subsection{RAICAR-N on random sets of 5 subjects - 100 group ICA runs}
To promote diversity across the group ICA runs, as discussed in section \ref{choosingL}, $L = 5$ subjects were drawn at random from the group of $N = 23$ subjects and submitted to a temporal concatenation based group ICA. This process was repeated $K = 100$ times and the resulting set of 100 group ICA maps were submitted to a RAICAR-N analysis. ICA components were sorted according to their reproducibility and $p$-values were computed for each ICA component. Please see Figure \ref{figure_pval_100groupof5}.

\begin{figure}[htbp]
\begin{center}
\includegraphics[width=6.0in]{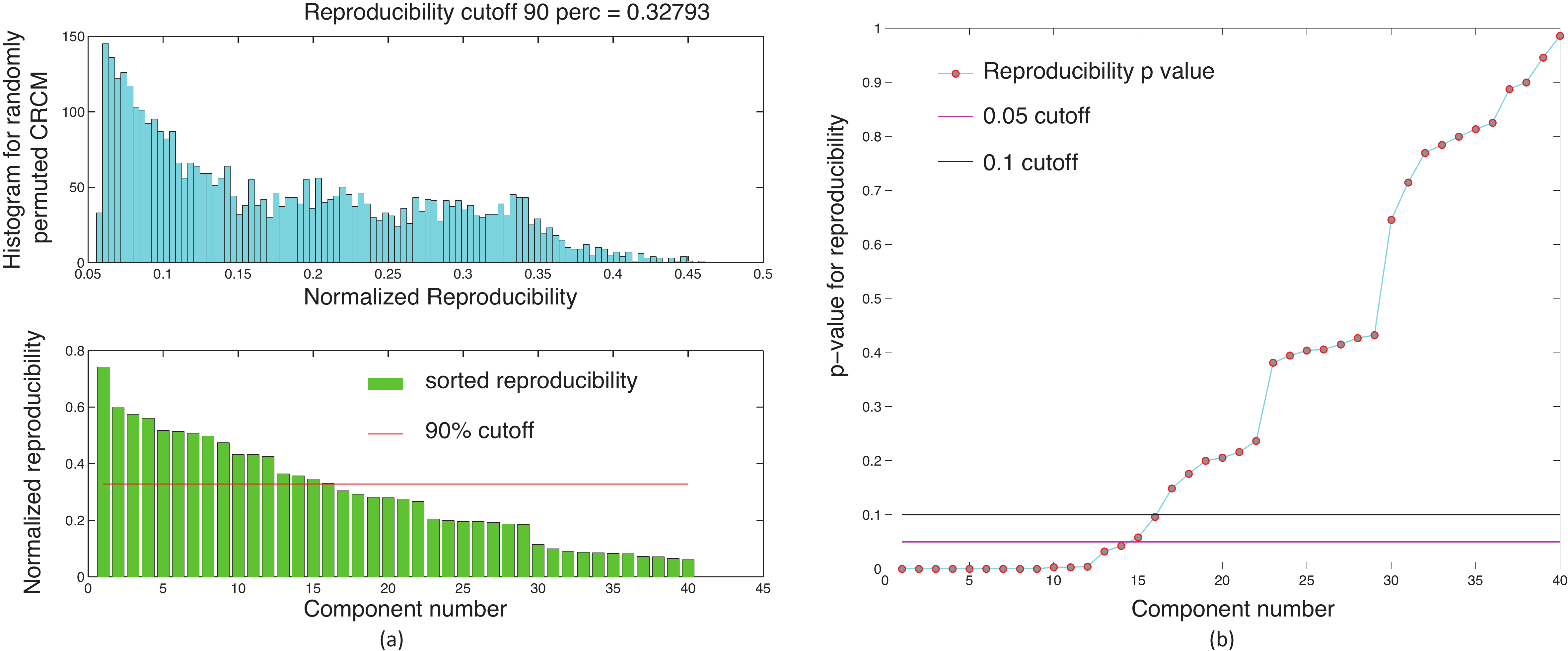}
\caption{$L=5$ subjects were randomly drawn from the set of $N=23$ subjects and submitted to a temporal concatenation based group ICA. This process was repeated $K=100$ times and the resulting ICA maps were submitted to a RAICAR-N analysis. Figure (a) shows the observed values of normalized reproducibility (bottom) as well as the "null" distribution of normalized reproducibility across $R=100$ simulations (top). Figure (b) shows the $p$-values for each IC along with the $0.05$ and $0.1$ cutoff lines.}
\label{figure_pval_100groupof5}
\end{center}
\end{figure}

We compared the reproducible RSNs from the single subject RAICAR-N analysis to the RSN maps reported in literature \citep{Beckmann:2005}. Please see Figure \ref{figure_top15maps_summary_100}.

\begin{figure}[htbp]
\begin{center}
\includegraphics[width=4.5in]{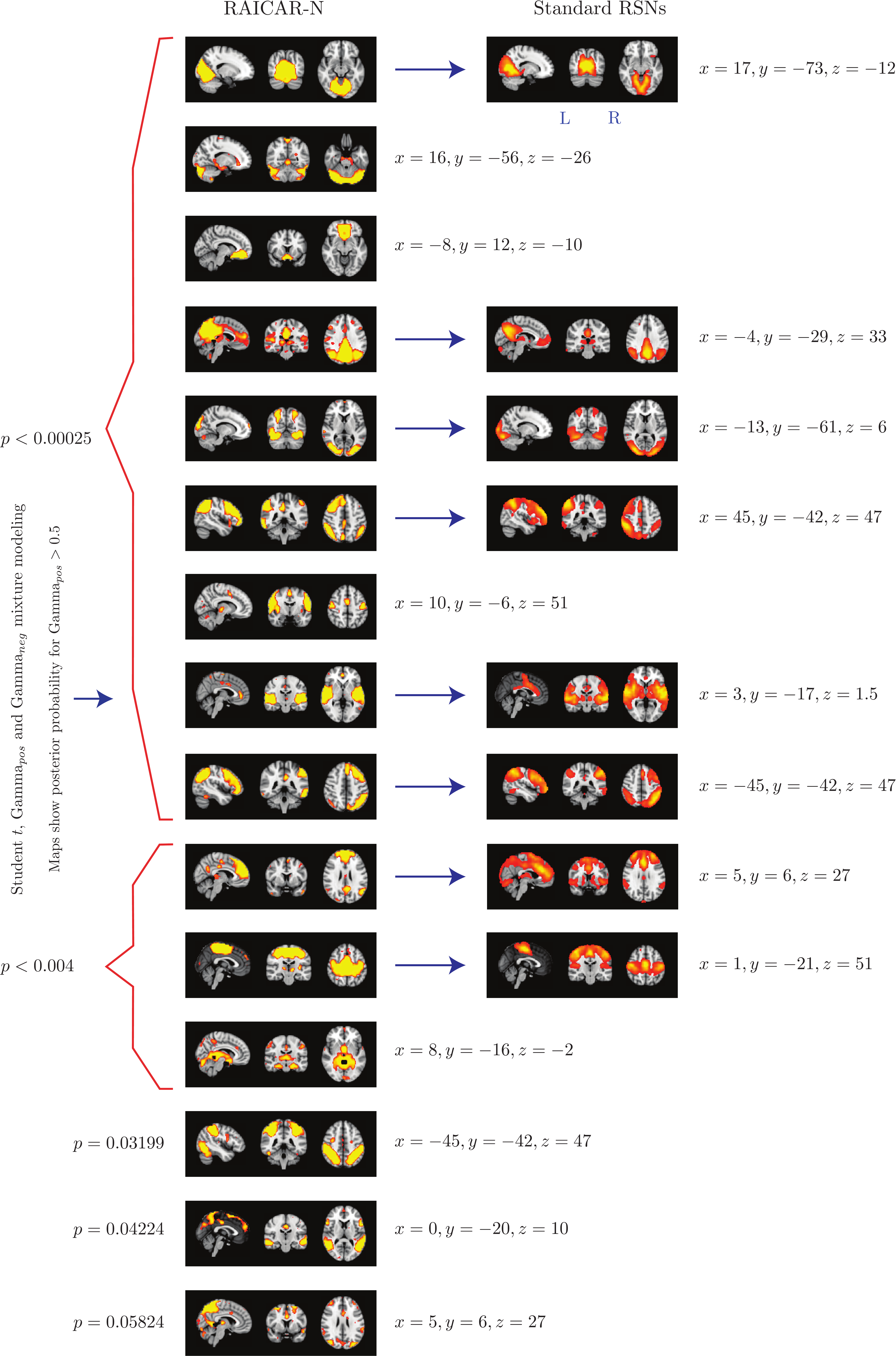}
\caption{The top 15 "reproducible" ICs from $K=100$ runs of $L=5$ subject group ICA RAICAR-N analysis compared with standard RSN maps reported in literature \citep{Beckmann:2005}. We are able to declare 8 "standard" RSNs as significantly reproducible at a $p$-value of $<0.05$. There are 6 other "non-standard" RSNs that can be declared as significantly reproducible at a $p$-value of $<0.05$ and 1 other "non-standard" RSN that achieves a $p$-value of $0.05824$.}
\label{figure_top15maps_summary_100}
\end{center}
\end{figure}

In summary, when 100 random 5 subject group ICA runs (from a population of 23 subjects) are combined using RAICAR-N: 
\begin{itemize}
\item We are able to declare 8 "standard" RSNs as significantly reproducible at a $p$-value $< 0.05$. 
\item There are 6 other "non-standard" RSNs that can be declared as significantly reproducible at a $p$-value $< 0.05$. 
\item There is 1 other "non-standard" RSN that achieves a $p$-value of 0.05824. 
\end{itemize}

\newpage

\section{Group comparison of ICA results}\label{raicar-n-across-groups}
In this section, we summarize the main approaches for group analysis of ICA results which can be broadly classified into two categories: (1) Approaches based on a single ICA run or no ICA run and (2) Approaches based on multiple ICA runs. To make things concrete, suppose we have two groups of subjects $A$ and $B$.

\subsection{Approaches based on a single group ICA run or no ICA run}

The main idea in these approaches is to use the results of a group ICA using all subjects to derive subject specific spatial maps for group comparison.
A typical sequence of steps is as follows:

\begin{enumerate}

\item The first step involves extraction of a set of template IC maps or a set of template mixing matrix time courses. This can be accomplished using two techniques:

	\begin{enumerate}
	
		\item \textbf{Group ICA based template IC maps or time courses:}\\
		A temporal concatenation based group ICA is run using data from all subjects in group $A$ and $B$. This usually involves two PCA data reductions. The first reduction is based on 
		a subject wise PCA decomposition \citep{Calhoun:2001} or an average PCA decomposition \citep{Beckmann:2005} as discussed in section \ref{group_ica}. The next reduction is based on PCA reduced temporally concatenated data. Subsequently, the group ICs and the dual PCA reduced mixing matrix time courses are estimated using an ICA algorithm. 
	
		\item \textbf{User supplied set of template IC maps:}\\
		The user supplies a set of spatial maps, perhaps corresponding to an ICA decomposition on an independent data set.
			
	\end{enumerate}

\item The next step either uses template IC maps or time courses.

	\begin{enumerate}
		
		\item \textbf{Template time course based approach:}\\
		First, the mixing matrix is PCA back projected and partitioned into subject specific sub matrices. Next, subject specific spatial maps corresponding to the group ICs are estimated via least-squares and a second PCA back projection is used to estimate the corresponding subject specific time courses. This is the approach proposed in \citep{Calhoun:2001}, which we will refer to as the group ICA back projection approach.
				
		\item \textbf{Template IC based approach:}\\
		First, spatial multiple regression using the template ICs as regressors is used against the original data of each subject to derive subject specific time courses corresponding to each template IC. Next, a second multiple regression using the subject specific time courses is used against the original data of each subject to derive subject specific spatial maps corresponding to each template IC. This approach called "dual-regression" has been proposed by \citep{Beckmann_dreg:2009}. A similar approach called fixed average spatial ICA (FAS ICA) had also been proposed earlier in \citep{Calhoun_dreg:2004}.	Both dual-regression and FAS ICA involve the first spatial regression stage, but dual-regression also includes a second temporal regression stage.	
	\end{enumerate}

	\item Once subject specific spatial maps and time courses corresponding to group ICs have been determined, they are entered into a random effects analysis for group comparison.

\end{enumerate}

\subsubsection{Advantages of single group ICA based approaches}

\begin{enumerate}
\item Much reduced computational load compared to multiple ICA based approaches.
\item Ability to take advantage of constrained spatial IC estimation across all subjects via group ICA.
\end{enumerate}

Please see section \ref{discussion} for discussion.

\subsection{Approaches based on multiple single subject or group ICA runs}

In these approaches results of multiple ICA runs in groups $A$ and $B$ are used for a between group analysis.
A typical sequence of steps is as follows:

\begin{enumerate}

\item The first step involves:

\begin{itemize}
\item running a separate single subject ICA for all subjects from groups $A$ and $B$ (possibly with multiple runs per subject) or 
\item running a set of group ICA runs across various sets of subjects separately, with each set containing subjects either from group $A$ or group $B$
\end{itemize}

\item The next step is to establish a correspondence between the ICs within and across groups. There are two main techniques of establishing this correspondence:

	\begin{enumerate}
		\item \textbf{Template based methods:}\\
		In these approaches, the user defines a template or a spatial map containing the network of interest. Examples of templates include a spatial map of the default mode network (DMN) derived from a separate ICA analysis, a spatial map from a separate PCA analysis, or even a binary mask defining the regions of interest. The template is then used to select from each run of ICA (single subject or group ICA) in each group ($A$ and $B$), an IC that best matches the template using a predefined metric such as spatial correlation coefficient or goodness of fit (GOF) \citep{Greicius:2004}. 
		
		\item \textbf{Template free methods:}\\
		These approaches do not need a pre-defined template from the user, but instead attempt to match or cluster all ICs simultaneously within and across groups. Examples of such approaches include self organizing group ICA (sogICA, \citep{Esposito:2005}) and RAICAR \citep{Yang:2008}. Each matched component or IC cluster includes one IC from each ICA run (single subject or group ICA) in each group ($A$ and $B$). 
	\end{enumerate}
	
\item Finally, the selected ICs in template based methods or ICs from a selected IC cluster/matched component in template free methods are then entered into a random effects group analysis (with repeated measures for multiple single subject ICA runs) for between group comparison.

\end{enumerate}

\subsubsection{Advantages of multiple ICA run approaches}

\begin{enumerate}
\item They account for both algorithmic and data set variability of ICA.
\item Group comparisons happen on true ICs i.e., optimal solutions for the ICA problem.
\end{enumerate}

Please see section \ref{discussion} for discussion.

\section{Discussion}\label{discussion}

As discussed in section \ref{identifiability}, in the noisy linear ICA model with isotropic diagonal Gaussian noise co-variance, for a given true model order, the mixing matrix and the source distributions are identifiable upto permutation and scaling. However, as pointed out in section \ref{run-to-run}, various factors prevent the convergence of ICA algorithms to unique IC estimates. These factors include ICA model not being the true data generating model, approximations to mutual information used in ICA algorithms, multiple local minima in ICA contrast functions, confounding Gaussian noise as well as variability due to model order over-estimation. A practical implication of these factors is that ICA algorithms converge to different IC estimates depending on how they are initialized and on the specific data used as input to ICA. Hence, there is a need for a rigorous assessment of reproducibility or generalizability of IC estimates. A set of reproducible ICs can then be used as ICA based characteristics of a particular group of subjects.

We proposed an extension to the original RAICAR algorithm for reproducibility assessment of ICs within or across subjects. The modified algorithm called RAICAR-N builds up a "null" distribution of normalized reproducibility values under a random assignment of observed ICs across the $K$ runs. This "null" distribution is used to compute reproducibility $p$-values for each observed matched component from RAICAR. An objective cutoff such as $p < 0.05$ can be used to detect "significantly reproducible" components. This avoids subjective user decisions such as selection of the number of clusters in ICASSO or the reproducibility cutoff in RAICAR or a cutoff on intra cluster distance in sogICA.

\subsection{Results for publicly available rsfMRI data}
We applied RAICAR-N to publicly available $N=23$ subject rsfMRI data from \url{http://www.nitrc.org/}. We analyzed the data in 2 different ways:

\begin{enumerate}
\item $n_C=40$ ICs were extracted for each of the $N=23$ subjects. The $K=23$ single subject ICA runs were subjected to a RAICAR-N analysis (after registration to standard space). 

In single subject ICA based RAICAR-N analysis (see Figures \ref{figure_pval_23subj} - \ref{figure_top8maps_summary}), we are able to declare 6 out of the 8 ICs reported in \citep{Beckmann:2005} (which used group ICA) as "reproducible" (4 ICs have $p$-values $< 0.05$ and 2 ICs have $p$-values $< 0.06$). This is consistent with the 5 reproducible RSNs reported in \citep{DeLuca:2006} using single subject ICA analysis. 

\item $L=5$ subjects were randomly drawn from $N=23$ subjects to create one group of subjects which was subjected to a group ICA analysis in which $n_C=40$ components were extracted. This process was repeated $K=50$ or $100$ times and the resulting group ICA runs were subjected to a RAICAR-N analysis.

 In group ICA based RAICAR-N analysis (see Figures \ref{figure_pval_50groupof5} - \ref{figure_top15maps_summary_100}), we are able to declare all 8 components reported in \citep{Beckmann:2005} as "reproducible" (at $p < 0.05$). Some of the ICs detected as "reproducible" in the group ICA based RAICAR-N on human rsfMRI data are not shown in \citep{Beckmann:2005} but do appear in the more recent paper \citep{Smith:2009}. RAICAR-N results for $K = 50$ are almost identical to those for $K = 100$ suggesting that $K = 50$ runs of group ICA are sufficient for a RAICAR-N reproducibility analysis.

\end{enumerate}

\subsection{Single subject ICA vs Group ICA}
Based on our results, it appears that single subject ICA maps are less reproducible compared to group ICA maps as illustrated in Figures \ref{figure_pval_23subj} and \ref{figure_pval_50groupof5}. A single subject ICA based analysis is more resistant to subject specific artifacts. On the other hand, a group ICA based analysis makes the strong assumption that ICs are spatially identical across subjects. If this assumption is true, group ICA takes advantage of temporal concatenation to constrain the ICs spatially across subjects thereby reducing their variance. Hence, when there are no gross artifacts in individual rsfMRI data sets, group ICA is expected to be more sensitive for reproducible IC detection. As seen in Figures \ref{figure_top8maps_summary} and \ref{figure_top15maps_summary}, our results agree with this proposition. All ICs declared as "reproducible" in the single subject based RAICAR-N analysis continue to remain "reproducible" in the group ICA based RAICAR-N analysis. 

\subsection{How should subjects be grouped for group ICA?}
This raises the question of how the subjects should be grouped together for individual group ICA runs in preparation for RAICAR-N. If all $N$ subjects are used in all group ICA analyses then there is no diversity in the individual group ICA runs. In this case, a RAICAR-N analysis will capture algorithmic variability due to non-convexity of ICA objective function but not dataset variability. Hence, our conclusions might not be generalizable to a different set of $N$ subjects.

Another option is to randomly select $L$ subjects out of $N$ for each group ICA run and submit the resulting $K$ group ICA runs to RAICAR-N. In this case, we will account for both algorithmic and data set variability via a RAICAR-N analysis. In other words, we will be able to determine those ICs that are "reproducible" across different sets of $L$ subjects and across multiple ICA runs. A key question is: How should we choose $L$ and $K$? In section \ref{choosingL}, we proposed a simple method to determine the number of subjects $L$ to be used in a single group ICA run out of the $N$ subjects - the key idea is to form groups with enough "diversity". Multiple such group ICA runs can then be submitted to a RAICAR-N analysis for reproducibility assessment. Clearly, the larger the value of $N$, the larger the value of $L$. Hence, increasing the number of subjects $N$ in a study will allow us to make conclusions that are generalizable to a larger set of $L$ subjects. Also, conclusions generalizable to $L_1$ subjects are expected to hold for $L_2 > L_1$ subjects but not vice versa.

\subsection{RAICAR-N for group comparisons of \textit{reproducible} ICs}\label{raicar-n-across-groups}
In the present work, our focus was on enabling the selection of reproducible ICs for a given single group of subjects. However, RAICAR-N can be extended for between group analysis of reproducible components as well. Before we describe how to do so, it is useful to discuss other approaches for group analysis of RSNs described in section \ref{raicar-n-across-groups}. Suppose we have two groups of subjects $A$ and $B$.

\subsubsection{Discussion of single group ICA based approaches}

\begin{enumerate}

\item Subject specific maps corresponding to group ICA maps derived using ICA back projection or dual regression are not true ICs, i.e., they are not solutions to an ICA problem.

\item These approaches do not account for either the algorithmic or the data set variability of an ICA decomposition. The single group ICA decomposition will contain both reproducible and non-reproducible ICs, but there is no systematic way to differentiate between the two.

\item Both dual regression and ICA back projection using data derived IC templates are circular analyses. First, group ICA using \textbf{all} data is used to derive template IC maps or template time courses. Next least-squares based ICA back projection or dual regression using a subset of the same data is used to derive subject specific maps and time courses corresponding to each IC. Thus model $1$ (group ICA) on data $\mathcal{D}$ is used to learn an assumption $\mathcal{A}$ (template IC maps or template time courses) that is then used to fit model $2$ (dual regression or ICA back projection) on a subset of the \textbf{same} data $\mathcal{D}$. This is circular analysis \citep{Kriegeskorte:2009, Vul:2010}.

It is easy to avoid circular analysis in a dual regression approach via cross-validation. For example, one can split the groups $A$ and $B$ into two random parts, a "training" set and a "test" set. First, the "training" set can be used to derive template IC maps using group ICA. Next, the "training" set based template IC maps can be used as spatial regressors for dual regression on the "test" set. Alternatively, the template ICs for dual regression can also come from a separate ICA decomposition on a independent data set unrelated to groups $A$ and $B$ such as human rsfMRI data. This train/test approach cleanly avoids the circular analysis problem. It is not clear how to use cross-validation for an ICA back projection approach since template time courses cannot be assumed to remain the same across ICA decompositions.

\item Subject specific structured noise is quite variable in terms of its spatial structure. Hence, a group ICA analysis cannot easily model or account for subject specific structured noise via group level ICs. Consequently, subject specific spatial maps in ICA back projection or dual regression will have a noise component that is purely driven by the amount of structured noise in individual subjects. On the other hand, a single subject ICA based analysis can accurately model subject specific structured noise via single subject ICs.
\end{enumerate}

\subsubsection{Discussion of multiple ICA run approaches}

\begin{enumerate}

\item \citep{Zuo:2010} report that using different sets of template ICs in template based methods using spatial correlation such as \citep{Harrison:2008} can result in the selection of different ICs in individual ICA runs. This is not surprising since IC correspondence derived from template based methods \textbf{does} depend on the particular template used. This is similar to a seed based correlation analysis being dependent on the particular seed ROI used. It is worth noting that template free approaches such as sogICA and RAICAR do not rely on any template.

\item \citep{Cole:2010} state that individual runs across subjects (or groups of subjects) can be quite variable in terms of the spatial structure of the estimated ICs. For example, \citep{Cole:2010} point out that an IC might be apparently split into two sub-components in some subjects but not others. The real problem is that the same model order could lead to over-fitting in some subjects (or groups of subjects) but not in others. Hence, the observed differences in a group comparison might be biased by the unknown difference in the amount of over-fitting across groups $A$ and $B$. 

As described in \ref{run-to-run}, over-fitting can lead to the phenomenon of component "splitting" in ICA. This is not limited to single subject ICA but can also occur in group ICA. For instance, \citep{Zuo:2010} report the "default mode" network as split into three sub networks using group ICA and note that component "splitting" can also reflect functional segregation or hierarchy within a particular IC and is not necessarily a consequence of model order overestimation in every case.

Over-fitting can be correctly accounted for by a reproducibility analysis. This is because we expect the real and stable non-Gaussian sources to be reproducible across multiple ICA runs (algorithmic variability) and across different subjects or groups of subjects (data set variability). 

\end{enumerate}

If we want the results of a between group ICA analysis to be generalizable to an independent group of subjects then we must account for both the algorithmic and data variability of ICA. We propose to modify RAICAR-N for enabling between group comparisons of "reproducible" ICs as follows:

\begin{enumerate}
\item Enter multiple within and across subject (or within and across sets of subjects) ICA runs for groups $A$ and $B$ into a RAICAR analysis. Perform the RAICAR component matching process across groups $A$ and $B$.
\item Use RAICAR-N to compute reproducibility $p$-values separately for group $A$ and $B$ for each matched component across groups $A$ and $B$. 

\item Only ICs that are \textit{separately} reproducible in both groups $A$ and $B$ and that are \textit{maximally} similar to each other are used for between group comparisons.
\end{enumerate}

To summarize, a RAICAR-N analysis:
\begin{itemize}
\item[\Checkmark]  can be applied for "reproducible" component detection either within or across subjects in any component based analysis - not necessarily ICA. For instance, a set of PCA maps across subjects can be submitted to a RAICAR-N analysis.

\item[\Checkmark]  is simple to implement and accounts for both algorithmic and data set variability of an ICA decomposition.

\item[\Checkmark] avoids any user decisions except the final $p$-value cutoff which can be objectively pre-set at standard values such as $0.05$.

\item[\Checkmark] can be extended to enable comparisons of reproducible ICs between groups $A$ and $B$.
\end{itemize}

\section{Conclusions}\label{conclusion}
Multiple group ICA runs using groups of subjects with enough "diversity" can be used to account for the run-to-run variability in ICA algorithms both due to the non-convex ICA objective function as well as across subjects data variability. These group ICA runs can be subjected to a RAICAR-N "reproducibility" analysis. RAICAR-N enables the objective detection of "reproducible components" in any component based analysis of fMRI data such as ICA and can also be used for a between group comparison of "reproducible" ICs.

\section*{Acknowledgements}
We gratefully acknowledge financial support from the Pain and Analgesia Imaging and Neuroscience (P.A.I.N) group, McLean Hospital, Harvard Medical School, Belmont MA, USA under the grants K24NS064050 (DB) and R01NS065051 (DB). We would also like to thank Dr. Christian Beckmann for making the IC image files from his 2005 paper \citep{Beckmann:2005} available to us.

\bibliographystyle{plainnat}
\bibliography{bibliography_raicar_N_arxiv.bib}

\section{Figure Legends}
\addtocounter{figure}{-13}
\begin{figure*}[htbp]
\begin{center}
\includegraphics[width=4.5in]{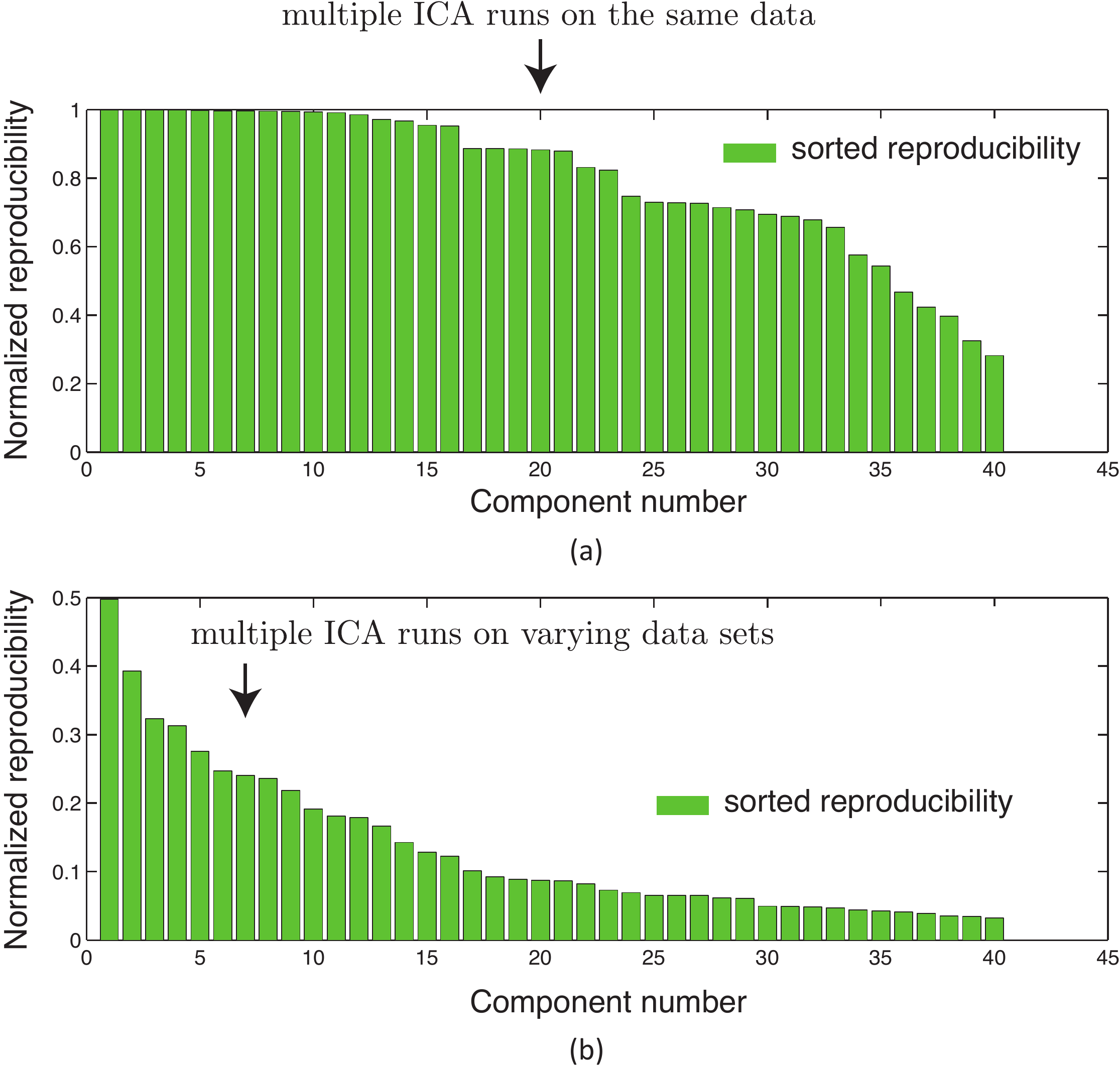}
\caption{Figure illustrates the variation in normalized reproducibility from RAICAR depending on whether the input to RAICAR is (a) Multiple ICA runs on single subject data or (b) Multiple ICA runs across subjects. Notice that the normalized reproducibility is much lower for across subjects analysis compared to within subject analysis.}
\label{figure1}
\end{center}
\end{figure*}

\begin{figure}[htbp]
\begin{center}
\includegraphics[width=6.2in]{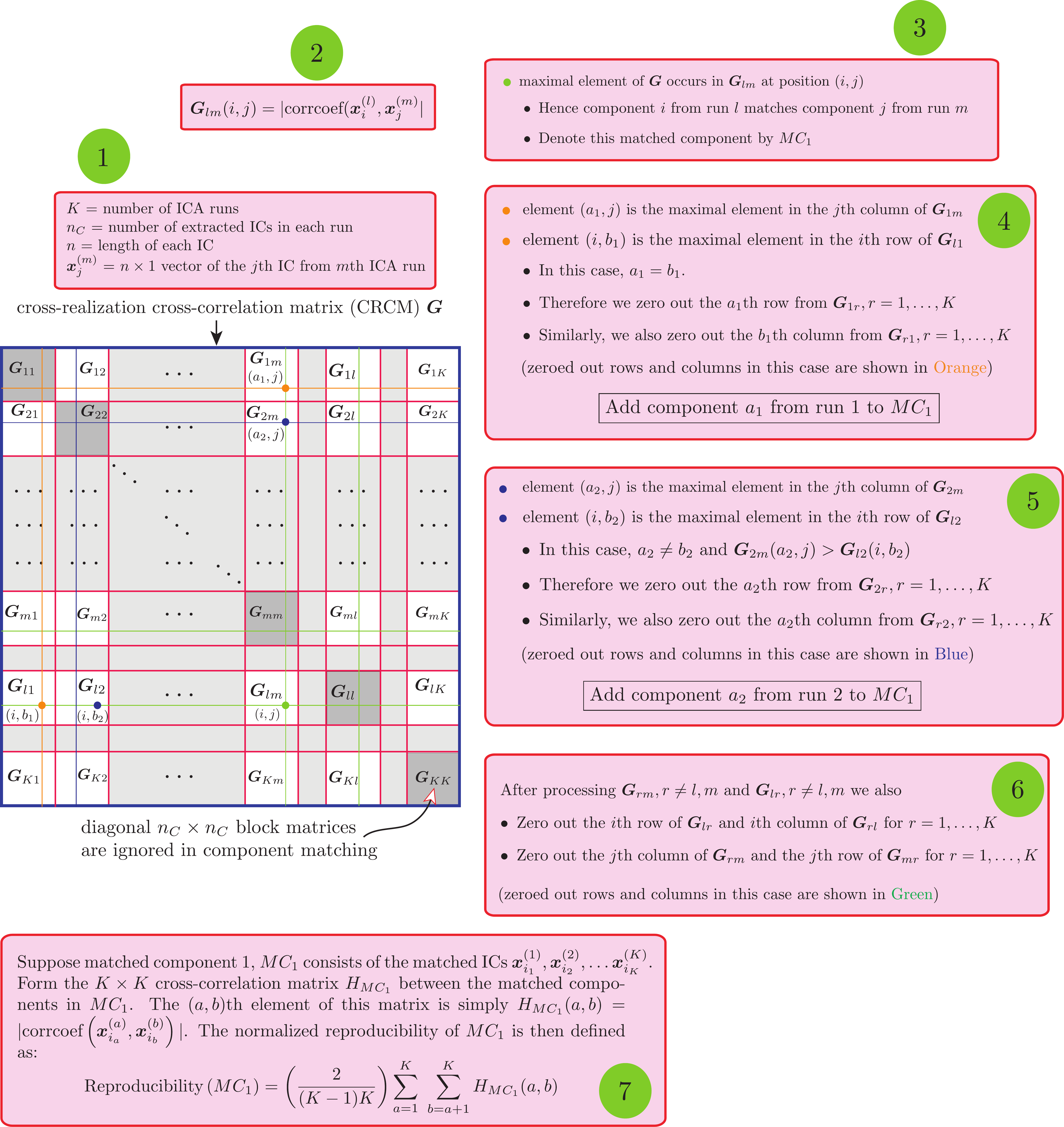}
\caption{Pictorial depiction of the original RAICAR algorithm \cite{Yang:2008}. The ICA algorithm is run $K$ times with each run producing $n_C$ ICs. $\matr{G}$ is a $K \times K$ block matrix with elements $\matr{G}(l,m) = \matr{G}_{lm}$ where $\matr{G}_{lm}$ is the $n_C \times n_C$ absolute spatial cross-correlation matrix between ICs from runs $l$ and $m$. The numbered green circles indicate the sequence of steps in applying RAICAR to a given data set. Our definition of normalized reproducibility in box 7 averages un-thresholded correlation coefficients thereby avoiding the selection of a correlation coefficient threshold prior to averaging.}
\label{figure2}
\end{center}
\end{figure}

\begin{figure}[htbp]
\begin{center}
\includegraphics[width=4.1in]{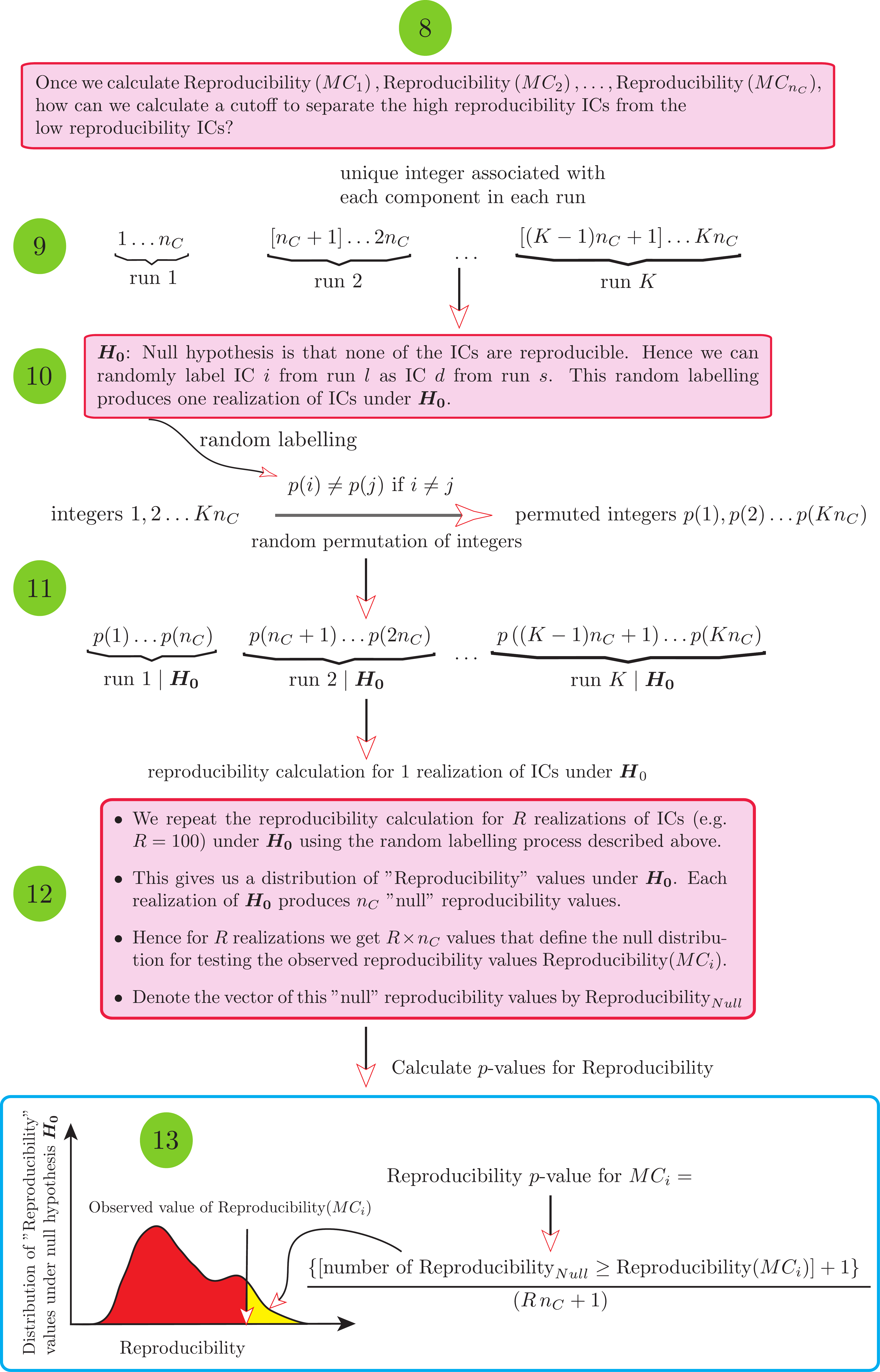}
\caption{Pictorial depiction of the process for generating a "null" distribution in RAICAR-N. Our "null" hypothesis is: "$\mathbf{H_0}$: None of the ICs are reproducible. Hence, we can randomly label IC $i$ from run $l$ as IC $d$ from run $s$". Therefore we randomly split the $K n_C$ ICs across $K$ runs into $K$ parts and run the RAICAR algorithm on each set of randomly split ICs. This gives us a set of "null" reproducibility values which can be used to compute $p$-values for the observed reproducibility of ICs in the original RAICAR run. The green circles indicate the sequence of steps for generating the "null" distribution after the steps in Figure \ref{figure2}.}
\label{figure3}
\end{center}
\end{figure}

\begin{figure}[htbp]
\begin{center}
\includegraphics[width=6.3in]{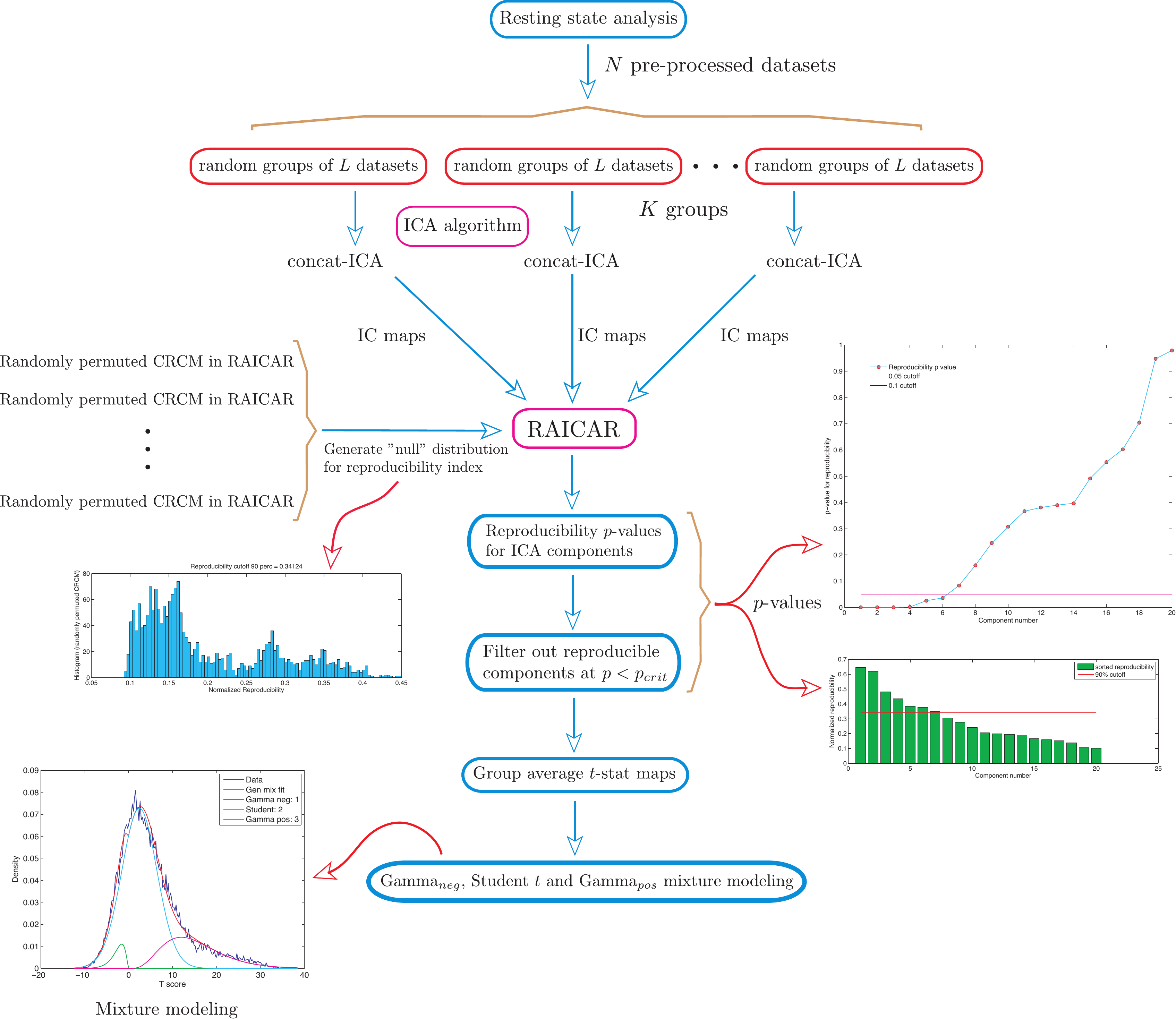}
\caption{Flowchart for a group ICA based RAICAR-N analysis. The $N$ single subject data sets are first pre-processed and subsequently bootstrapped to create $K$ groups, each group containing $L$ distinct subjects. Each group of $L$ subjects is submitted to a temporal concatenation group ICA analysis. The resulting IC maps (either raw ICs or ICs scaled by noise standard deviation) are subjected to a RAICAR analysis. The cross-realization cross correlation matrix (CRCM) is randomly permuted multiple times: $\matr{G} \rightarrow \matr{G}(\vect{g},\vect{g})$ where $\vect{g}$ is a random permutation of integers from $1,\ldots,K n_C$. The permuted CRCMs are subjected to a RAICAR analysis to generate a realization of reproducibility values under the "null" hypothesis. The computed "null" distribution of reproducibility values is used to assign $p$ values to the observed reproducibility of the original RAICAR run. Finally, reproducible ICs are averaged using a random effects analysis and the resulting $t$-statistic images are subjected to Gamma$_{neg}$, Student $t$ and Gamma$_{pos}$ mixture modeling.}
\label{figure4}
\end{center}
\end{figure}

\begin{figure}[htbp]
\begin{center}
\includegraphics[width=3.0in]{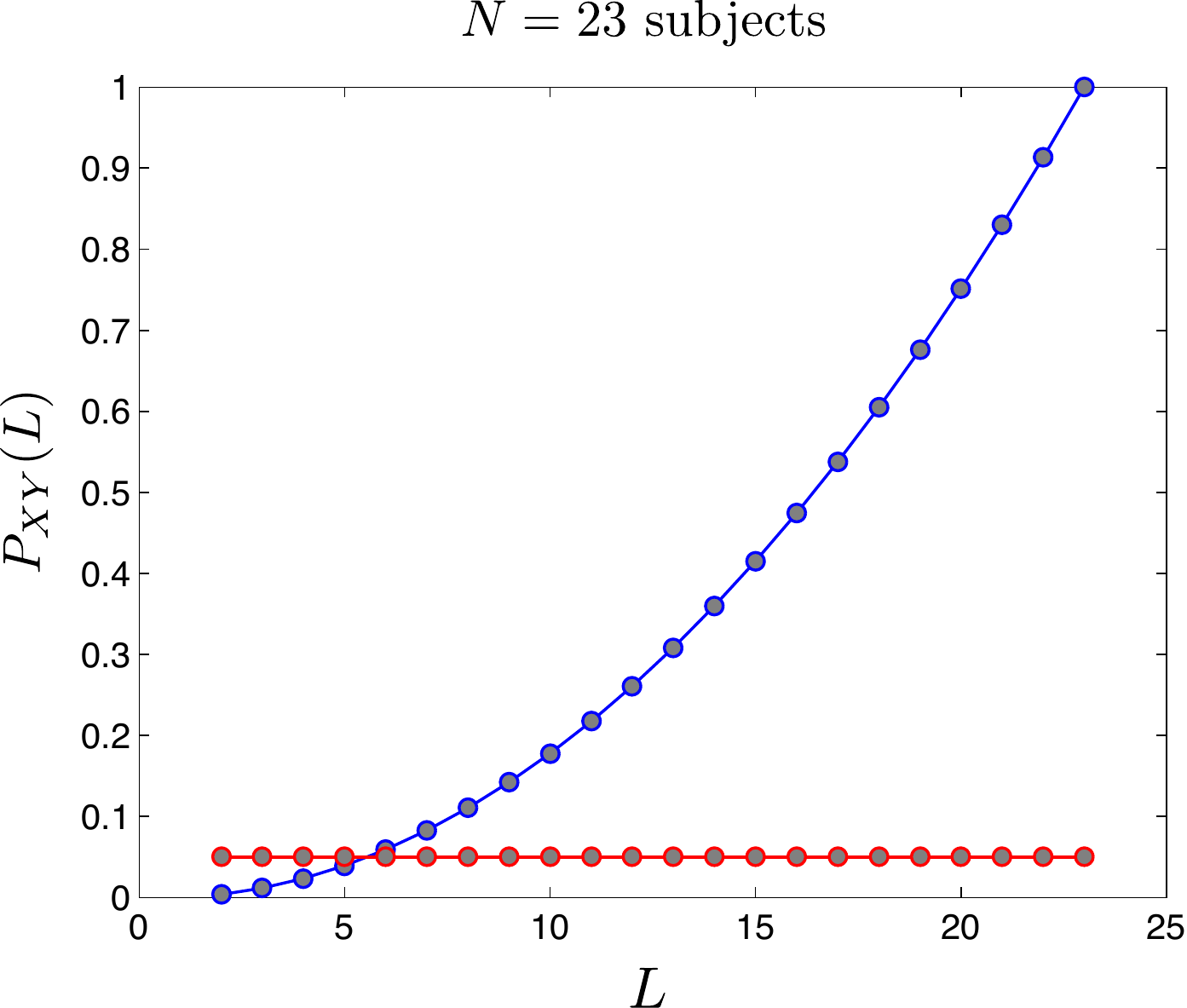}
\caption{Figure shows a plot of $P_{XY}(L)$ vs $L$ for $N = 23$ in \textcolor{blue}{blue}. The \textcolor{red}{red} line shows the $\alpha_{max} = 0.05$ cutoff. The largest value of $L$ for which $P_{XY}(L) \le 0.05$ is $L = 5$.}
\label{figure5}
\end{center}
\end{figure}

\begin{figure}[htbp]
\begin{center}
\includegraphics[width=6.0in]{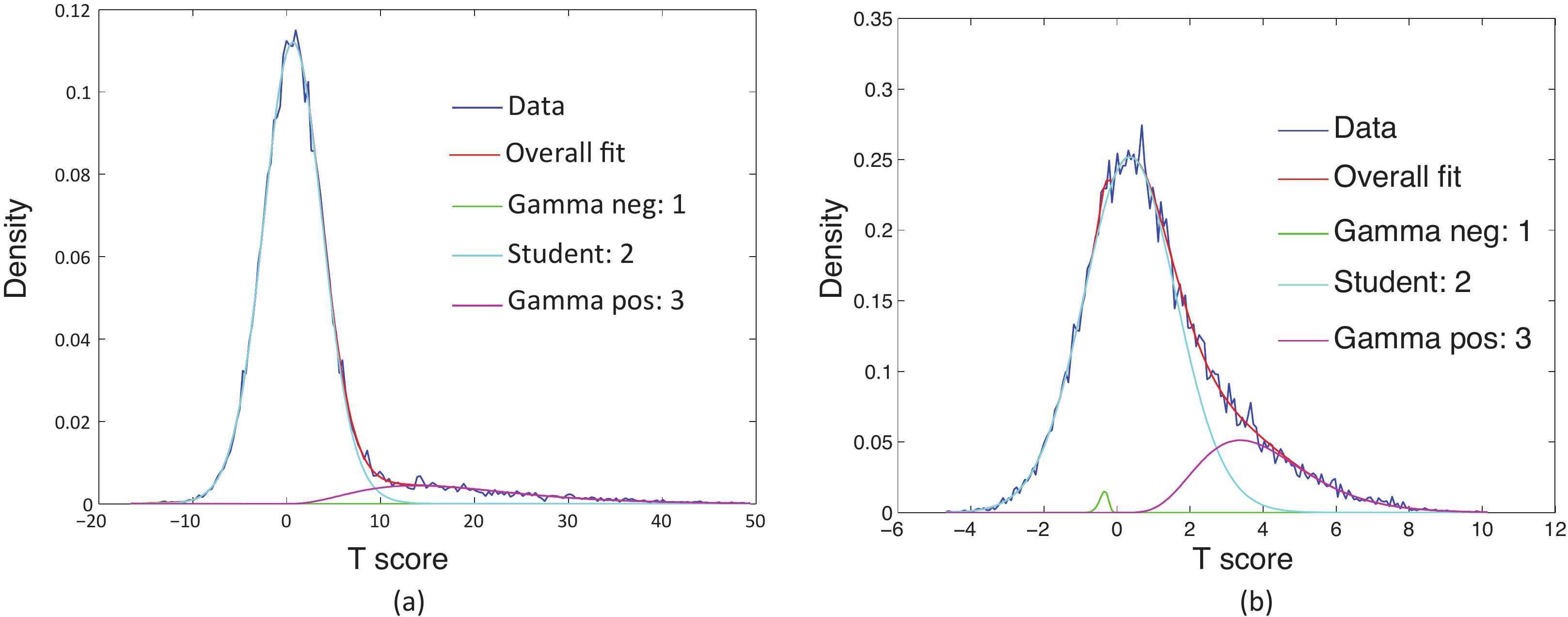}
\caption{Examples of displaying non-Gaussian spatial structure using a Student $t$, Gamma$_{pos}$ and Gamma$_{neg}$ mixture model. Notice how the Gamma$_{neg}$ density is driven to near $0$ class fraction in the absence of significant negative non-Gaussian structure.}
\label{figure6}
\end{center}
\end{figure}

\begin{figure}[htbp]
\begin{center}
\includegraphics[width=6.0in]{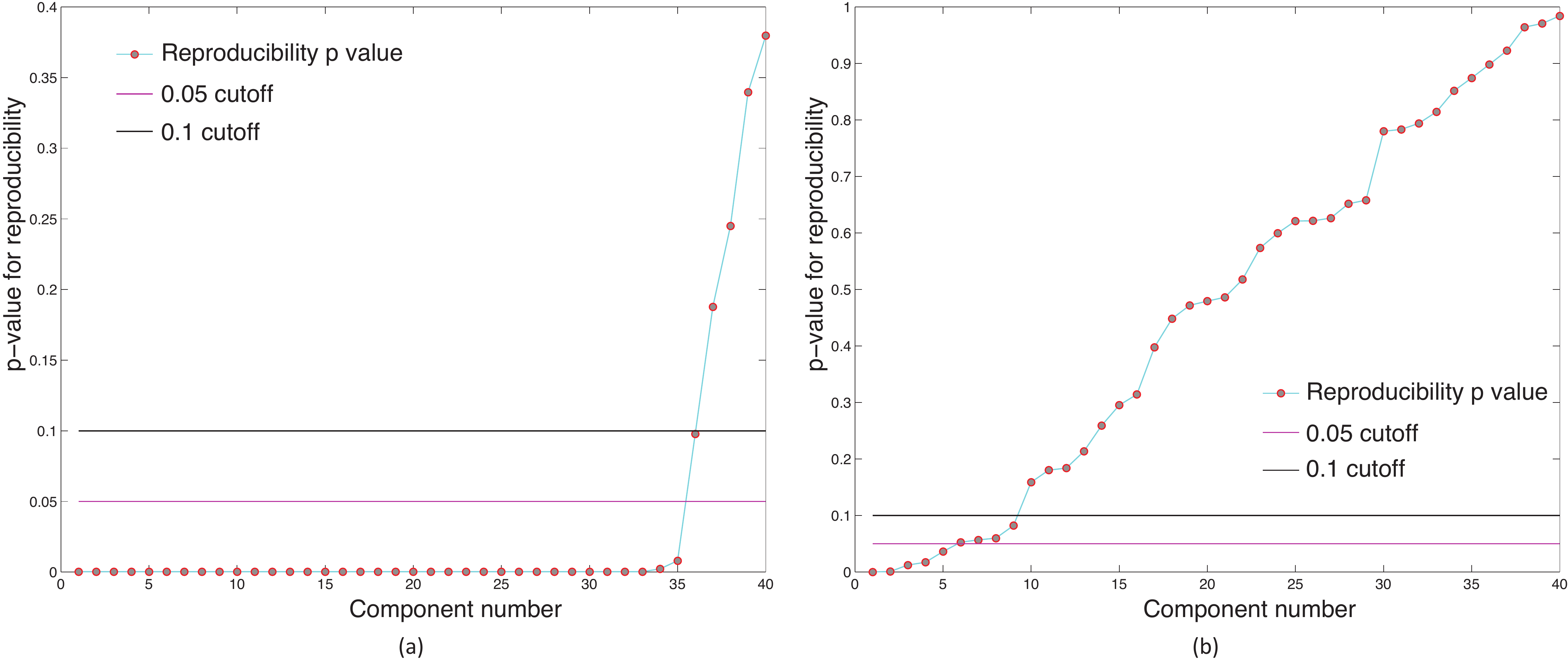}
\caption{$p$-value cutoffs for within and across single subject analysis using RAICAR-N. This figure illustrates the intuitive fact that within subject ICA runs are much more reproducible compared to across subject ICA runs.}
\label{figure7}
\end{center}
\end{figure}

\begin{figure}[htbp]
\begin{center}
\includegraphics[width=6.0in]{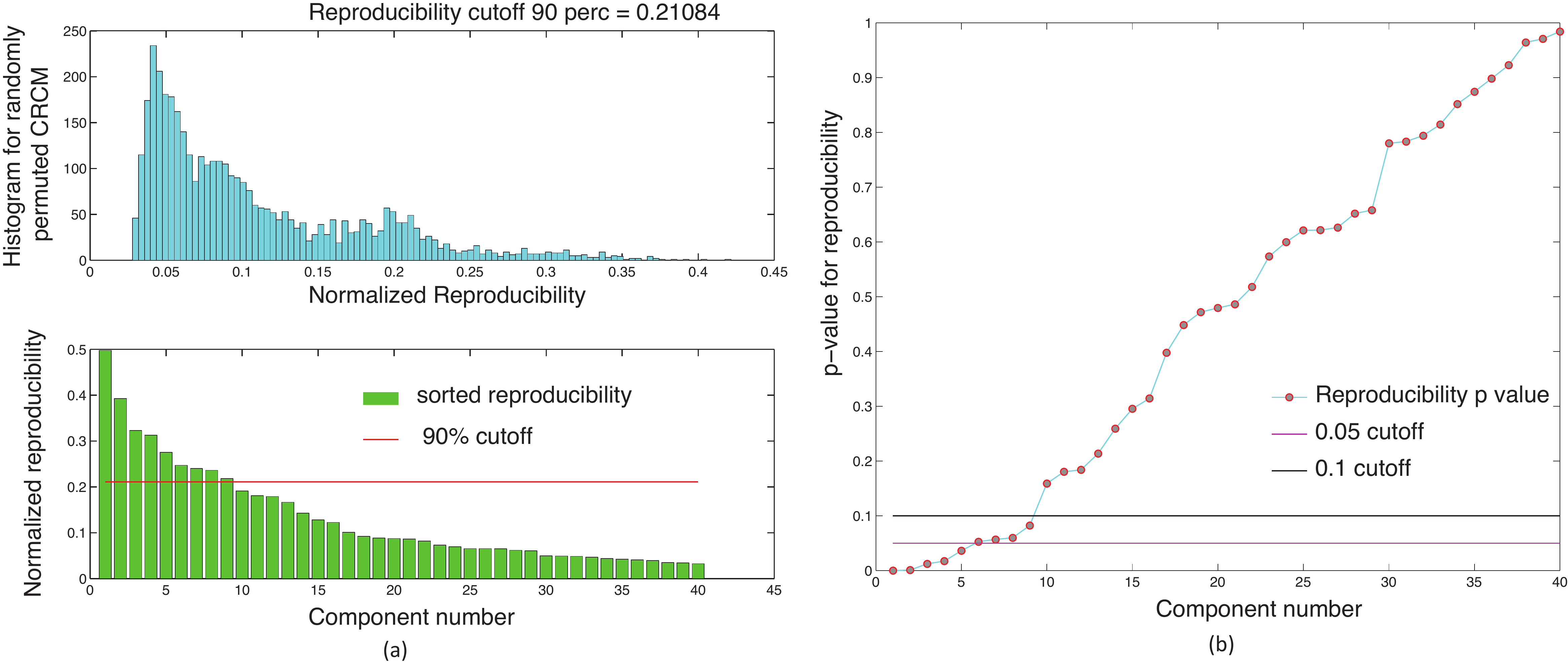}
\caption{Single subject rsfMRI ICA runs across 23 subjects were combined using a RAICAR-N analysis. Figure (a) shows the observed values of normalized reproducibility (bottom) as well as the "null" distribution of normalized reproducibility across $R=100$ simulations (top). Figure (b) shows the $p$-values for each IC along with the $0.05$ and $0.1$ cutoff lines.}
\label{figure_pval_23subj}
\end{center}
\end{figure}

\begin{figure}[htbp]
\begin{center}
\includegraphics[width=5.8in]{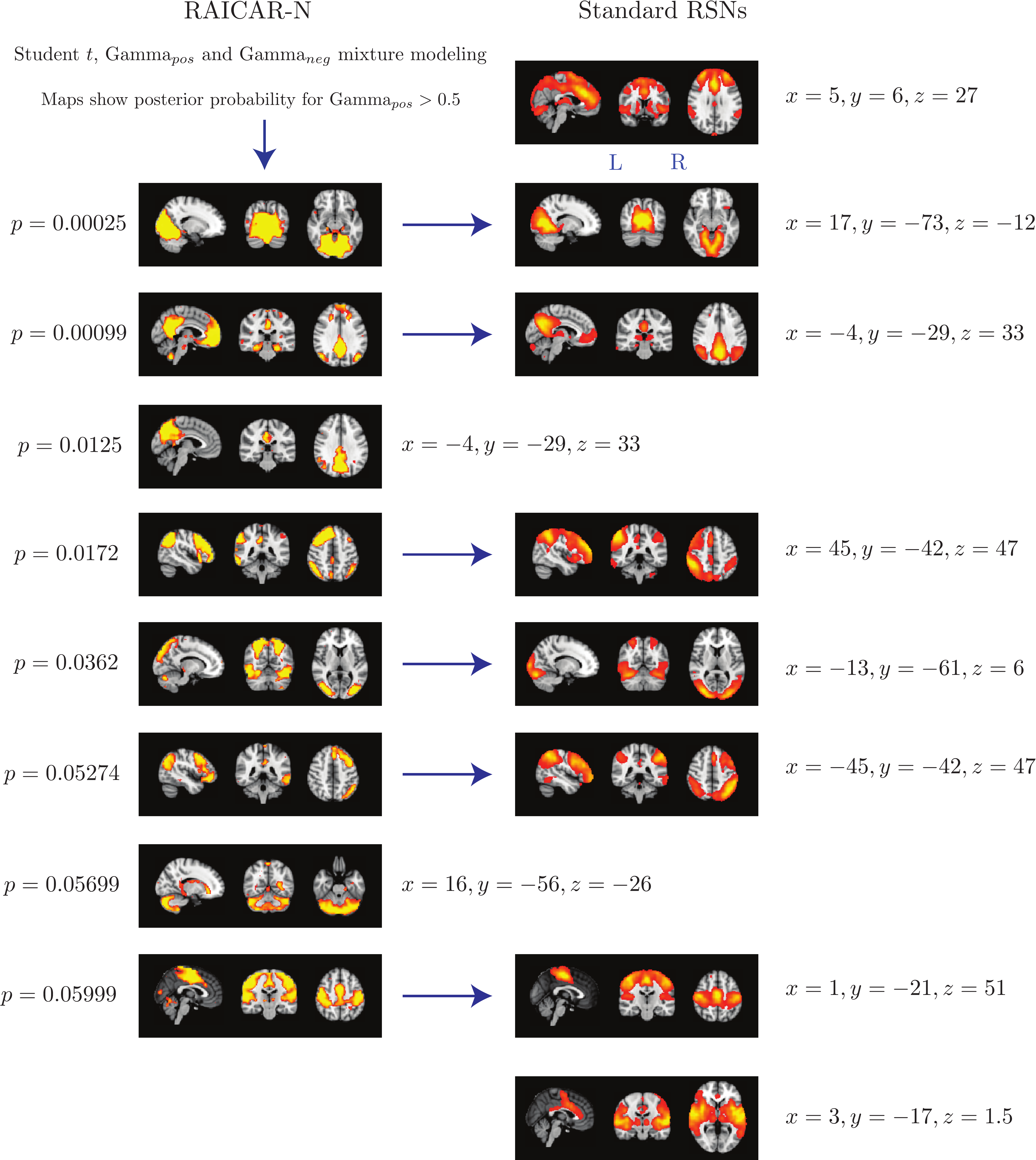}
\caption{The top 8 "reproducible" ICs from a RAICAR-N analysis on single subject ICA runs compared with standard RSN maps reported in literature \cite{Beckmann:2005}. We are able to declare 4 "standard" RSNs as significantly reproducible at a $p$-value $<0.05$. There are 2 other "standard" RSNs that achieve a reproducibility $p$-value between $0.05$ and $0.06$ as well as 2 "non-standard" RSNs that achieve $p$-values of $0.0125$ and $0.05699$ respectively. We also could not find 2 of the published RSNs in \cite{Beckmann:2005} as reproducible in single subject ICA runs.}
\label{figure_top8maps_summary}
\end{center}
\end{figure}

\begin{figure}[htbp]
\begin{center}
\includegraphics[width=6.0in]{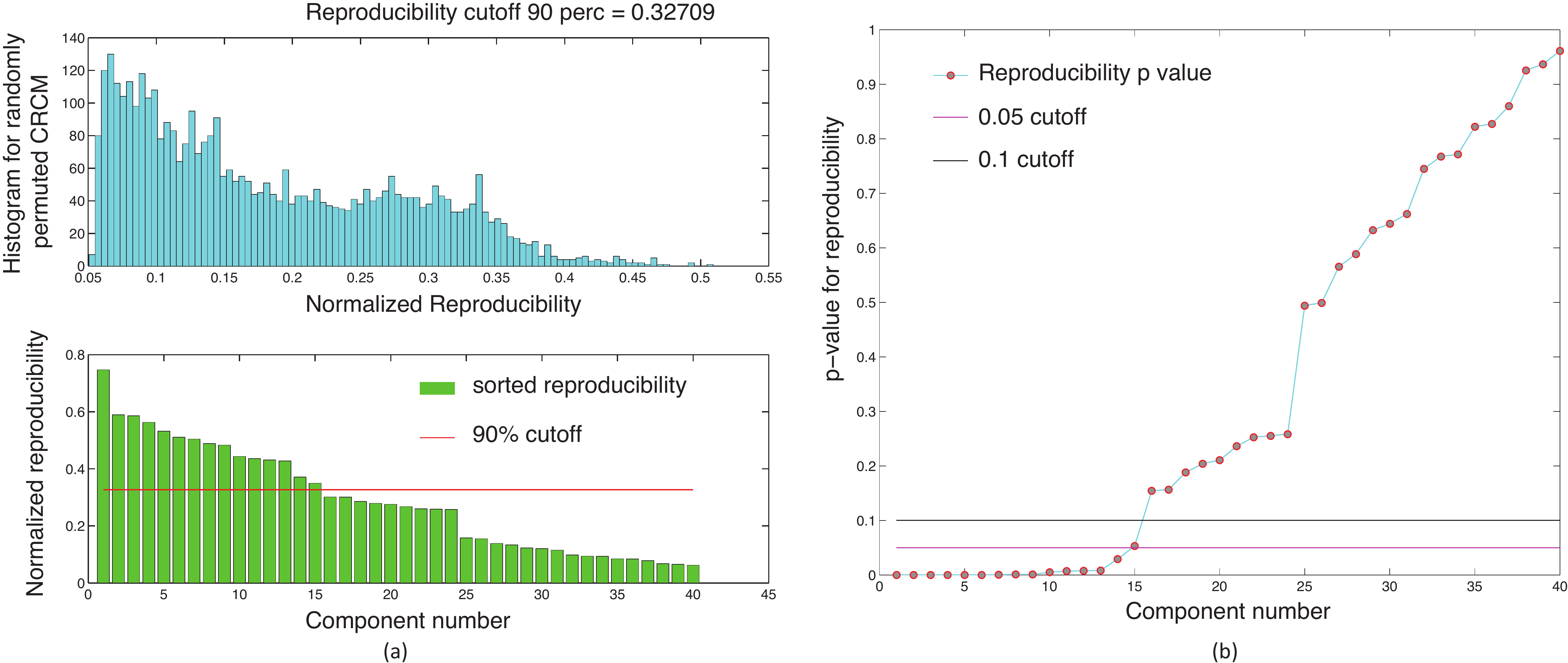}
\caption{$L=5$ subjects were randomly drawn from the set of $N=23$ subjects and submitted to a temporal concatenation based group ICA. This process was repeated $K=50$ times and the resulting ICA maps were submitted to a RAICAR-N analysis. Figure (a) shows the observed values of normalized reproducibility (bottom) as well as the "null" distribution of normalized reproducibility across $R=100$ simulations (top). Figure (b) shows the $p$-values for each IC along with the $0.05$ and $0.1$ cutoff lines.}
\label{figure_pval_50groupof5}
\end{center}
\end{figure}

\begin{figure}[htbp]
\begin{center}
\includegraphics[width=4.7in]{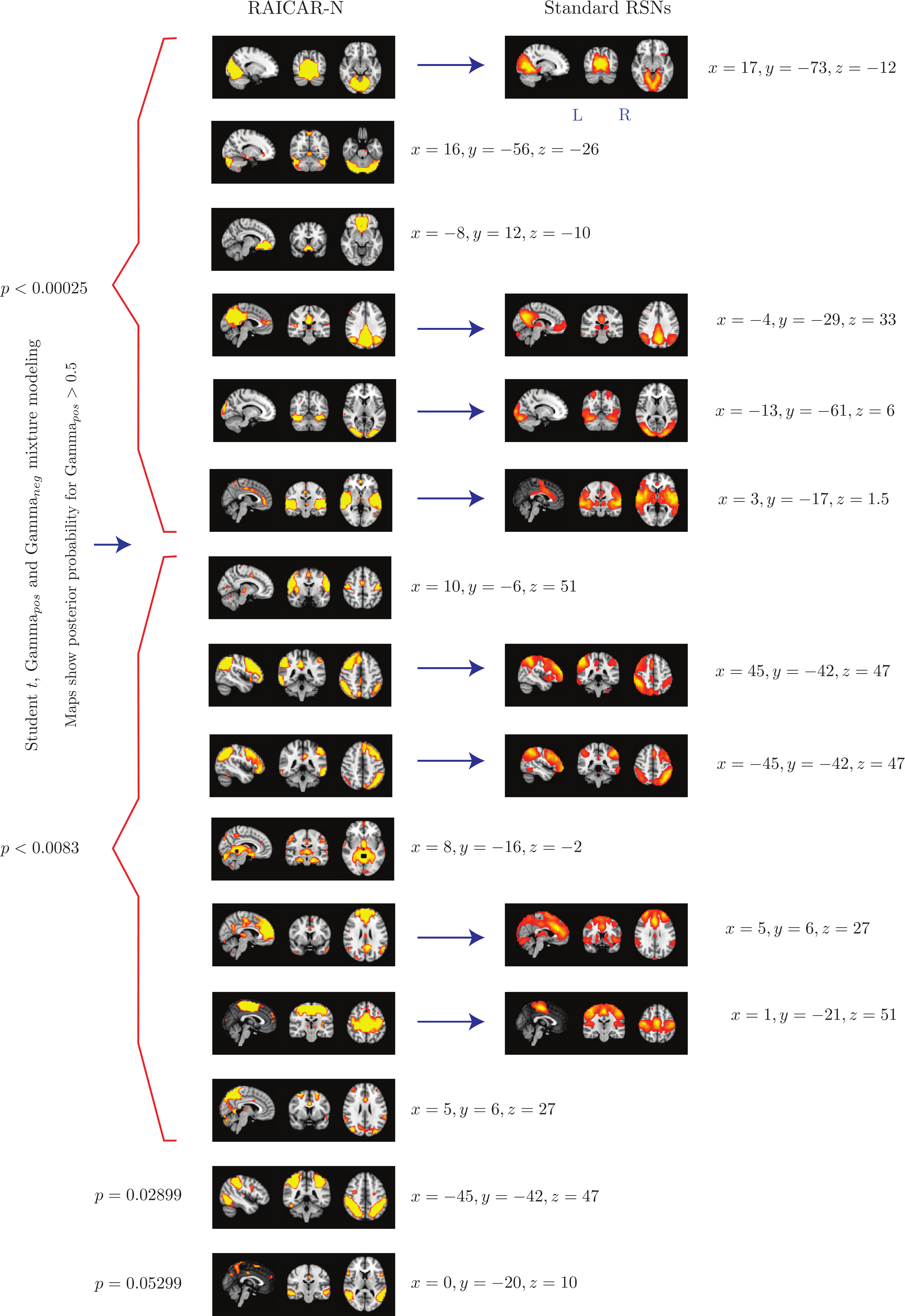}
\caption{The top 15 "reproducible" ICs from $K=50$ runs of $L=5$ subject group ICA RAICAR-N analysis compared with standard RSN maps reported in literature \cite{Beckmann:2005}. We are able to declare 8 "standard" RSNs as significantly reproducible at a $p$-value of $<0.05$. There are 6 other "non-standard" RSNs that can be declared as significantly reproducible at a $p$-value of $<0.05$ and 1 other "non-standard" RSN that achieves a $p$-value of $0.05299$.}
\label{figure_top15maps_summary}
\end{center}
\end{figure}

\begin{figure}[htbp]
\begin{center}
\includegraphics[width=6.0in]{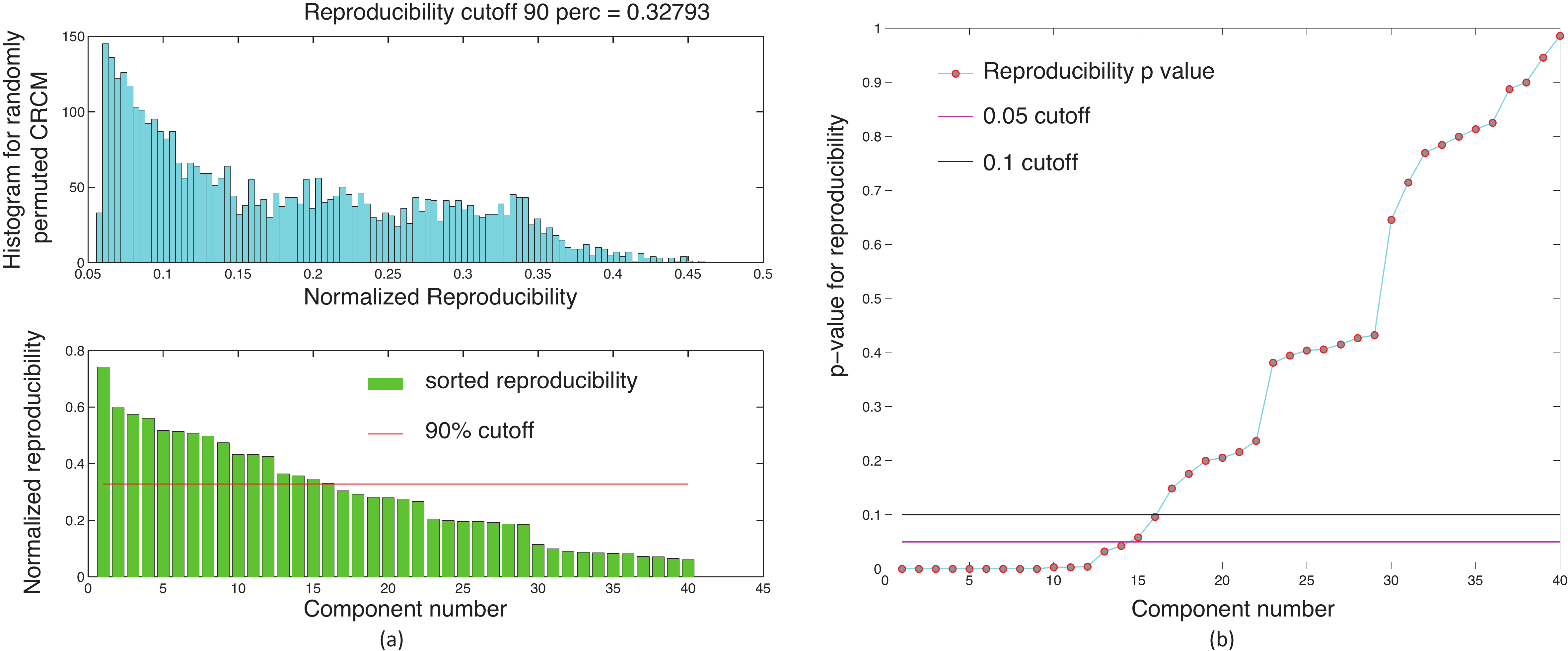}
\caption{$L=5$ subjects were randomly drawn from the set of $N=23$ subjects and submitted to a temporal concatenation based group ICA. This process was repeated $K=100$ times and the resulting ICA maps were submitted to a RAICAR-N analysis. Figure (a) shows the observed values of normalized reproducibility (bottom) as well as the "null" distribution of normalized reproducibility across $R=100$ simulations (top). Figure (b) shows the $p$-values for each IC along with the $0.05$ and $0.1$ cutoff lines.}
\label{figure_pval_100groupof5}
\end{center}
\end{figure}

\begin{figure}[htbp]
\begin{center}
\includegraphics[width=4.5in]{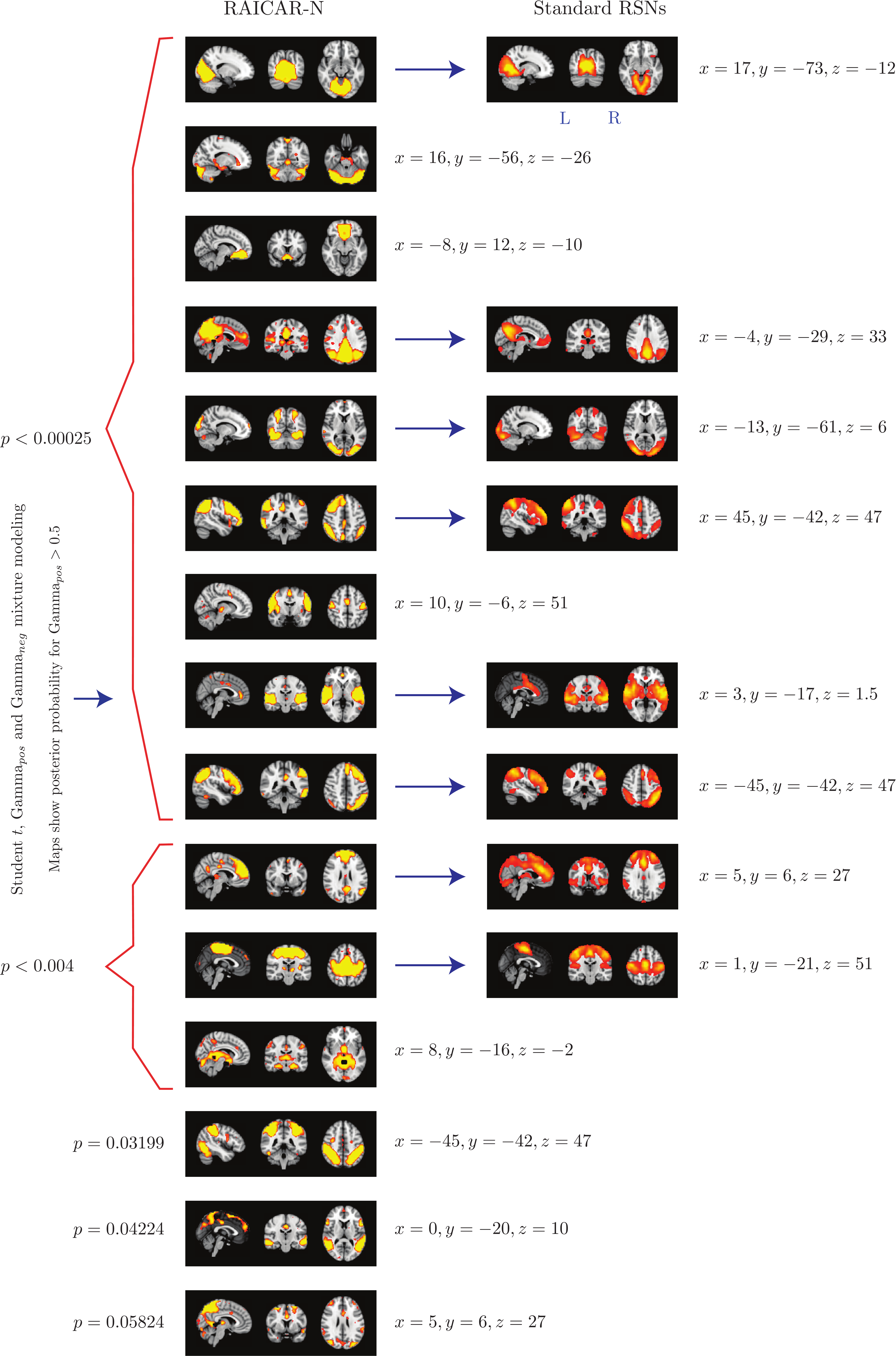}
\caption{The top 15 "reproducible" ICs from $K=100$ runs of $L=5$ subject group ICA RAICAR-N analysis compared with standard RSN maps reported in literature \cite{Beckmann:2005}. We are able to declare 8 "standard" RSNs as significantly reproducible at a $p$-value of $<0.05$. There are 6 other "non-standard" RSNs that can be declared as significantly reproducible at a $p$-value of $<0.05$ and 1 other "non-standard" RSN that achieves a $p$-value of $0.05824$.}
\label{figure_top15maps_summary_100}
\end{center}
\end{figure}

\end{document}